\documentclass[12pt]{extarticle}
\usepackage{amsmath,amsthm,amssymb,color,xcolor,mathtools,amsfonts,bm,chemfig,hhline,subcaption}
\usepackage[ pdfborderstyle={/S/U/W 1}]{hyperref}
\definecolor{MyBlue}{HTML}{210cac}
\definecolor{MyCiteColor}{HTML}{0099FF}
\definecolor{MyRed}{HTML}{3E186A}%{CC033C}
\hypersetup{%
  pdfborderstyle={/S/U/W 1},
  colorlinks=true,% hyperlinks will be coloured
  linkcolor=MyBlue,% hyperlink text will be green
  citecolor=MyCiteColor,
  urlcolor=MyRed,
  anchorcolor=MyBlue,
  linkbordercolor=MyRed,% hyperlink border will be red
}
\usepackage[shortlabels]{enumitem}
\usepackage{graphicx}
\usepackage{makecell}
\usepackage[final]{pdfpages}
\setboolean{@twoside}{false}
\usepackage{pdfpages}
\usepackage{tikz}
\usepackage{booktabs}% http://ctan.org/pkg/booktabs
\usepackage{colortbl}% http://ctan.org/pkg/colortbl
 \definecolor{Ftitle}{RGB}{11,46,108}
\definecolor{line}{RGB}{87,39,117}
\colorlet{tableheadcolor}{Ftitle!25} % Table header colour = 25% gray
\colorlet{tablerowcolor}{gray!10} % Table row separator colour = 10% gray

\usetikzlibrary{shapes,arrows,spy,positioning,decorations,automata}
\usepackage{caption}
\usepackage{subcaption}
\usepackage{mathrsfs}
\usepackage{scalefnt}
\usepackage{verbatim}
\oddsidemargin \evensidemargin
\makeatletter
\newcommand{\eqnum}{\refstepcounter{equation}\textup{\tagform@{\theequation}}}
\makeatother

\newtheorem{theorem}{Theorem}
\numberwithin{theorem}{section}
\newtheorem{proposition}[theorem]{Proposition}

\newtheorem{definition}[theorem]{Definition}

\newtheorem{remark}[theorem]{Remark}
\newtheorem{example}[theorem]{Example}

\newcommand{\RR}{\mathbb{R}}
\newcommand{\QQ}{\mathbb{Q}}
\newcommand{\PP}{\mathbb{P}}
\newcommand{\CC}{\mathbb{C}}
\newcommand{\ZZ}{\mathbb{Z}}

 \date{}

\definecolor{bl1}{HTML}{1F558A}%{204073}
\definecolor{pur1}{HTML}{52196D}%98d2d9
\definecolor{mag1}{HTML}{2AD0F1}%{3BCDFF}

\title{Complexity of Model Testing for Dynamical Systems with Toric Steady States}

\author{Michael F. Adamer and Martin Helmer}

\begin{document}

\maketitle

\begin{abstract} \noindent
In this paper we investigate the complexity of model selection and model testing for dynamical systems with toric steady states. Such systems frequently arise in the study of chemical reaction networks. We do this by formulating these tasks as a constrained optimization problem in Euclidean space. This optimization problem is known as a Euclidean distance problem; the complexity of solving this problem is measured by an invariant called the Euclidean distance (ED) degree. We determine closed-form expressions for the ED degree of the steady states of several families of chemical reaction networks with toric steady states and arbitrarily many reactions. To illustrate the utility of this work we show how the ED degree can be used as a tool for estimating the computational cost of solving the model testing and model selection problems.

\end{abstract}

\section{Introduction}

Dynamical systems with toric steady states \cite{PerezMillan,Craciun2009} are ubiquitous in the modeling of natural phenomena. While our analysis will focus on examples arising from systems biology, the techniques used could be applied to study any dynamical system with toric steady states (see Definition \ref{def:DynSysWToircSteadyStates}).
%Of particular interest is the mathematical analysis of chemical concentration dynamics, formulated by non-linear ordinary differential equations \cite{Feinberg1987,Feinberg1988,Joshi2014,ConradiShiu2017,Salazar2009,Gross2016b}.
The analysis of chemical reaction networks forms a vital part of systems biology research \cite{Feinberg1987,Feinberg1988,Salazar2009,Kitano2002}. Our goal is to study chemical reaction networks with mass action kinetics for which the differential equations governing chemical concentration dynamics are polynomial \cite{Feinberg1987,Feinberg1988}. This restriction will allow us to apply ideas and algorithms from algebraic geometry to study several key features of chemical reaction networks. %Our analysis will focus on chemical reaction network models arising in systems biology. 

Due to the inherent complexity of the biological world it is often unknown which models best capture the biological mechanism. Therefore, many candidate models are often constructed to focus on a particular aspect of a biological system. When a set of candidate reaction mechanisms (i.e.~a set of models) has been developed, the optimal values of the parameters need to be identified.
%The task of finding a mathematical description of the biological mechanism often reduces to finding parameters which give an appropriate fit of the model to the measured data points \cite{Burnham2002,Chamberlin1965,Kirk2013}.
Hence, two important questions arising in modeling of biological systems are:
\begin{enumerate}
\item \textit{Model Selection:} Which mathematical model does most accurately describe the biological system?
\item \textit{Model Testing:} Is the chosen model capable of explaining the observed data?
\end{enumerate}
In this paper we will focus on the latter question of model testing by giving an upper bound on the complexity of finding the optimal parameter values. Our work is complimentary to the model selection approach presented in \cite{Gross2016a}. In \cite{Gross2016a} numerical algebraic geometry based algorithms for selecting the best fitting model were developed. In this paper we aim to quantify the computational complexity of the model selection task {\em without} solving the underlying equations. This ``model complexity'' will be inherent in all approaches which solve the underlying equations. We do, however, use numerical algebraic geometry tools to show the validity of our claims.

The main quantities needed for a practical answer to the model testing question are the steady state concentrations of a chemical reaction network which extremise the squared Euclidean distance to a given data point. Constrained Euclidean distance (ED) optimization problems of various types occur commonly in many applications. When the constraints are given by polynomial equations these problems may be solved using methods from algebraic geometry, techniques for this have been developed by several authors \cite{DHOST,OSS}. In mathematics the algebraic geometric ED problem have been studied in the contexts of low rank tensor and low rank matrix approximation, see for example \cite{OSS,SF17}. In systems biology, as discussed above, the ED problem has been used to study the model selection problem in \cite{Gross2016a}. Other areas where the algebraic geometric ED problem arises include phylogenetics \cite{CFM17}, computer vision \cite{THP15,FKO17}, signal processing \cite{CH15}, and sensor data analysis \cite{CCBAST,CNAS}.
%{\color{MyBlue} Add literature review (what's been done before, where has it been applied before, why is it useful in this context.}
%a solution of the optimization problem which seeks to find the nearest point in the model to the measured data point. 
In our setting the ED optimization problem consists of finding the solutions of a system of polynomial equations. The difficulty of solving this problem can be measured by an invariant called the Euclidean distance degree \cite{DHOST}. 

%In this paper we compute the (generic) Euclidean Distance degree, ${\rm gEDdegree}$ \cite{HS} of toric and non-toric models and show how such a computation can be used to for model selection and model testing. By the Euclidean Distance degree of the model we mean the Euclidean Distance degree of the steady state variety $V_{\mathfrak{N}}$ associated to a chemical reaction network $\mathfrak{N}$ with mass action kinetics. 

Our focus in this paper will be on computing exact formulas for the ED degree of the steady states of several chemical reaction networks with toric steady states. These formulas will be independent of the choice of rate constants $k$ and will be entirely determined by the graph of the chemical reaction network. The formulas will be found \textit{without} solving the associated polynomial system of critical equations. 

Below in Table \ref{tab:EDdegs} we tabulate the ED degrees for the chemical reaction models considered in this paper. The ED degrees provide the (relative) computational cost of solving the model testing and model selection problems for the different reaction networks. In particular, we see that three of the models, Processive Phosphorylation \cite{ProcPhos} (\S\ref{subsec:proc}), the Sequestration network \cite{ConradiShiu2017} (\S\ref{subsec:seq}) and the McKeithan model \cite{McKeithan1995} (\S\ref{subsec:mck}) have a small and constant ED degree relative to the number of reactions; meaning the testing and selection problems for these models can be solved in a practical time for an arbitrary number of reactions. On the other hand, Distributive Phosphorylation \cite{ProcPhos} (\S\ref{subsec:dist}) and Pore Forming models \cite{Lee2016} (\S\ref{subsec:pore}) have an ED degree which grows linearly with the number of reactions; this will in practice yield an approximately exponential growth in the run time for the computation of the solutions of the testing and selection problems for these models.

\begin{table}[!h]
\centering
\resizebox{.95\linewidth}{!}{
\begin{tabular}{@{} *8c @{}}
\toprule 
 \multicolumn{1}{c}     {\color{Ftitle} Processive (\S\ref{subsec:proc})}  &   {\color{Ftitle} Distributive (\S\ref{subsec:dist})} &   {\color{Ftitle} Sequestration (\S\ref{subsec:seq})} &   {\color{Ftitle} McKeithan (\S\ref{subsec:mck})} &   {\color{Ftitle} Pore (\S\ref{subsec:pore})}  \\ 
 \midrule 
 \color{line}28 & \color{line}$23N+5$ & \color{line}1& \color{line}6 & \color{line} $3N-2$\\ 
\bottomrule
 \end{tabular}}\vspace{1mm}
\caption{The ED degree for several families of chemical reaction models with varying number of reactions, given by $N$.} \label{tab:EDdegs}
 \end{table} 
 The paper is organized as follows. In \S\ref{section:Prelim} we review relevant background on chemical reactions networks and the related mathematical notions we will use to study them. In \S\ref{section:models} we compute the ED degree of the models listed in Table \ref{tab:EDdegs} for an arbitrary number of reactions. In \S\ref{section:Biology} we consider the biological interpretation of the results computed in \S\ref{section:models}. A summary of the work is given in \S\ref{section:Conclusion}. 
 
\section{Methods}\label{section:Prelim}
In this section we introduce the algebraic techniques we will use to investigate the complexity of the models discussed in \S \ref{section:models} and give several small results which will be used in their analysis. The required mathematical tools are given in \S\ref{section:mathPrelim1} and \S\ref{section:mathPrelim2} followed by a review of chemical reaction network theory in \S\ref{subsection:CRNT_Intro}. In \S\ref{subsection:realClosestPoints} we show that for the models considered here we are guaranteed to find a at least one biologically relevant (i.e.~positive and real) local minimum. We also give a brief overview of the computational tools available for solving the ED problem in \S \ref{subsect:CompMeth}. %The main results and computations are performed in \S\ref{section:models}.
%\subsection{Mathematical Preliminaries}\label{section:mathPrelim}
%\noindent{\em Projective Closures:} 
%\subsection{Overview of Methods}

Our objects of study, namely chemical reaction networks with mass action kinetics, correspond to dynamical systems defined by polynomial equations. We fix a chemical reaction network $\mathfrak{N}$ with mass action kinetics; this is a system of ordinary differential equations (ODEs) which is given by \begin{equation}
\dot{x}=f(x,k)
\label{eq:DynSys}
\end{equation}
where $f(x,k)$ is a system of polynomial equations in $x=(x_1,\dots,x_n)$ with the $k = (k_1,\dots, k_m)$ being fixed positive reaction constants. We consider the polynomials $f_i(x,k)\in \RR[x_1,\dots, x_n]$.
In biological systems the variables $x = (x_1,\dots,x_n)$ could for example represent the concentrations of the chemical species in the network $\{X_1,\dots,X_n\}$ and the parameters $k = (k_1,\dots,k_m)$ are the reaction rates of each chemical reaction in the network. When measuring the concentrations experimentally one is often only able to measure the steady state concentrations of species but not the reaction rates. Therefore, we only consider the steady states of \eqref{eq:DynSys} in this paper.
With this notation we can define the so-called \textit{steady state variety} as the algebraic variety \begin{equation}
V_\mathfrak{N}=V(f)\subset \CC^{n},\label{eq:SteadyStateVar}
\end{equation} which is the (complex) vanishing set of the system of polynomial equations $f_1(x,k)=\cdots =f_n(x,k)=0$.% $x \in \CC^n$. 

% \textbf{(\color{pur1}(Maybe tweak this) In biological applications the values of the model variables will represent chemicals reacting in a biological system, in such situations each reactant, i.e.~each model variable, must be present. This constraint means that in practice we are primarily interested in positive real steady states.}
Fix a list of rate constants $k=(k_1,\dots,k_m)$ in \eqref{eq:SteadyStateVar}. A main goal of chemical reaction network theory is to analyze steady state behavior of \eqref{eq:DynSys} and of central importance are steady states in which every single chemical species has positive concentration. That is we wish to study points in $V_{\mathfrak{N}}$ which are also in $(\mathbb{R}_{>0})^n=\left\lbrace (x_1,\dots, x_n)\in \mathbb{R}^n\; | \; x_i>0 \; \forall i \right\rbrace$. To this end we define the \textit{non-zero closure} of $V_{\mathfrak{N}}$ to be the variety ${V_{\mathfrak{N}}^{\neq 0}}=\overline{V_{\mathfrak{N}}\cap (\CC^*)^n}$, where the Zariski closure is taken in $\CC^n$. For a given variety $V_\mathfrak{N}$ in $\CC^n$ taking the non-zero closure, $V_\mathfrak{N}^{\neq 0}$, of $V_\mathfrak{N}$ has the effect of removing any irreducible component $W$ of $V_{\mathfrak{N}}$ such that every point in $W$ has at least one zero coordinate. %In particular we are interested in the semi-algebraic set $V_{\mathfrak{N}}\cap (\mathbb{R}_{>0})^n$ within $V_{\mathfrak{N}}\cap (\CC^*)^n$ however we will see that the later object, that is the non-zero closure of $V_{\mathfrak{N}}$, is more amenable to study. 

We can now frame the study of the two problems of model selection and model testing in relation to the steady state variety \cite{Gross2016b}.
% In chemical reaction network modeling, like in other forms of modeling, an important question is: how closely does a given model match observed data? We will refer to this problem as the model testing problem. More precisely, g
For the model testing problem we take a chemical reaction network $\mathfrak{N}$ with steady state variety $V_{\mathfrak{N}}$ and observed steady state data $u\in \RR^n$. We wish to test whether there exists a point $v\in V_{\mathfrak{N}}$ which is within some distance $\epsilon$ of our observed data. If such a point exists, then our model describes the observed data, i.e.~we wish to test if there exists a $v\in V_{\mathfrak{N}}$ such that $\left\Vert v-u \right\Vert <\epsilon$, where $\epsilon > 0$.

A related problem is the one of model testing \cite{Gross2016b}; given some observed data point $u$ and a collection of candidate models $\mathfrak{N}^{(\ell)}$, where $(\ell)$ denotes the $\ell^{th}$ model, we would like to know which model most closely approximates the data,  i.e.~for which model is the value of $\left\Vert v^{(\ell)}-u \right\Vert <\epsilon$ minimized for some point $v^{(\ell)}\in  V_{\mathfrak{N}^{(\ell)}}$? To solve both these problems we must compute the set of points $v\in V_{\mathfrak{N}}$ which minimize the expression $\left\Vert v-u \right\Vert$ for some data $u\in \RR^n$. We refer to such an optimization problem as an Euclidean distance (ED) problem \cite{DHOST}. The main goal of this paper is to apply the concept of the ED problem \cite{DHOST,HS} to chemical reaction networks.

Formally, we seek to find the points $v\in V_\mathfrak{N}\subset \RR^n$ which minimize the (weighted) Euclidean distance between the observed data $u$ and the model values $v$, given a data point $u=(u_1,\dots,u_n)\in \RR^n$ with weights $\lambda = (\lambda_1,\dots,\lambda_n)$. %{\color{MyBlue} MA:Can we give a more intuitive reason of why we weigh the ED? They are just there to ensure the ED degree is generic, right?... MH: Yes, otherwise the ED degree might be lower and we can't compute it combinatorially.... so it makes the most sense mathematically to set up the problem that has the nice combinatorial answer, with the understanding that the number is an upper bound (since we do things like take the projective closure it will always be an upper bound anyway)}
More precisely, for $\lambda_i\in \RR_{>0}$, we consider the following constrained optimization problem \begin{equation}
{\rm Minimize}\;\;\; d^2 = \sum_{i=1}^n \lambda_i (u_i-v_i)^2 \; \rm{subject \; to \;} v\in V_\mathfrak{N} .\label{eq:projectedOp}
\end{equation}
In practice one often needs to find all local solutions to \eqref{eq:projectedOp} to find the global minimum of the ED problem. Hence, the difficulty of solving the optimization problem \eqref{eq:projectedOp} for generic $u$ and generic $\lambda$ using algebraic methods is governed by the number of  complex critical points, which are the solutions of the critical equations.
%Explicitly, a critical point of \eqref{eq:projectedOp} is a smooth point in $V_\mathfrak{N}\subset \CC^n$ which is a solution to the critical equations. 
The critical equations are generated by taking the $n$-dimensional gradient of \eqref{eq:projectedOp} and they define an algebraic variety.

Fixing a choice of rate constants $k$ specifying a steady state variety $V_\mathfrak{N}$ (as in \eqref{eq:SteadyStateVar}) we define the \textit{(generic) Euclidean distance degree} of $V_\mathfrak{N}$, written as ${\rm EDdegree}(V_\mathfrak{N})$ as the number of complex critical points of \eqref{eq:projectedOp}. This number will be the same for any generic choice of $\lambda$ and $u$. For special choices of $\lambda$ there may be fewer critical points associated to \eqref{eq:projectedOp}, in such cases we will write ${\rm EDdegree}_{\lambda}(V_\mathfrak{N})$; similarly for special choices of $u$. For any choice of $\lambda,u$ (even non-generic choices), the ${\rm EDdegree}(V_\mathfrak{N})$ will be an upper bound on the number of critical points of \eqref{eq:projectedOp} (see \cite{DHOST}), that is $${\rm EDdegree}_{\lambda,u}(V_\mathfrak{N})\leq {\rm EDdegree}(V_\mathfrak{N}).$$

%\begin{equation}
%{\rm Minimize}\;\;\; \sum_{i=1}^n \lambda_i (u_i-y_i)^2 \; \rm{subject \; to \;} y\in V_\mathfrak{N} \label{eq:Op}
%\end{equation} for
%The \textit{(generic) Euclidean distance degree} of $V_\mathfrak{N}$, written ${\rm gEDdegree}(V_\mathfrak{N})$, is the number of complex critical points associated to the optimization problem \ref{eq:projectedOp} for generic $\lambda$ and generic $u$.
%{\color{MyBlue} I cannot decipher Anne's comment here: "-but you might explain that you are r(e)m(o)v(in)g ... comps". MH: I think this comment might also be about the reaction coefficients, i.e.~that we are removing them, but I think we have this mostly resolved now, maybe need to tweak a little, it was hard not to be technical about that... I need to think if there is a less technical way to write the justification on pg.~6.}
%In practice, to find the number of complex critical points of the optimization  variety $V$ the associated variety of critical points problem of computing the 
As above let $V_\mathfrak{N}^{\neq0}$ denote the non-zero closure of the steady state variety $V_\mathfrak{N}$ (i.e.~$V_\mathfrak{N}^{\neq0}$ is the result of removing irreducible components of $V_\mathfrak{N}$ with zero coordinates) associated to a chemical reaction network $\mathfrak{N}$. In what follows we will study the (generic) Euclidean distance degree of $V_\mathfrak{N}^{\neq 0}$. This number will provide a (reasonably sharp) estimate for the computational complexity of solving the model testing problem associated to \eqref{eq:projectedOp}. In particular the ED degree measures the difficulty of finding and representing \textit{all} solutions (and hence all real solutions) to the the ED problem \eqref{eq:projectedOp}. Additionally knowledge of the number of expected solutions to the ED problem \eqref{eq:projectedOp} (i.e.~knowing the ED degree) could be used to aid the design of specialized symbolic or numeric methods to solve these problems rapidly.  %It should be noted that while a number of classical numerical methods could be applied to solve the system \eqref{eq:projectedOp} these methods do not guarantee that all solutions will be found, i.e.~real solutions could be missed. and to quickly certify solutions found
% Unless explicitly stated we will use the nomenclature gEDdegree and EDdegree interchangeably. 

\subsection{Steady State Varieties in Projective Space}\label{section:mathPrelim1}
For chemical reaction networks the steady state variety, and hence the associated variety defined by the critical equations of \eqref{eq:projectedOp}, are objects in an affine space such as $\RR^n$ or $\CC^n$. To effectively derive the exact formulas for the ED degree presented below we will need to slightly change the ambient space. To intuitively understand the need for this consider, for example, the intersection of a parabola and a line in the real affine space $\RR^2$. In an affine space such an intersection may be empty, i.e.~the parabola could be above or below the line. The possibility of two such curves failing to intersect in an affine ambient space such as $\CC^n$ makes it impractical to compute exact expressions for the number of points in the intersection of curves and surfaces in these spaces without resorting to direct computational methods (such as computing a Gr\"obner basis of the ideal of critical points). To avoid this problem and to allow for the derivation of exact expressions for the ED degree without using the equations for the variety of critical points we will primarily work in an ambient projective space, $\PP^n$. In a projective space we are assured that, for example, a parabola and a line will have a non-empty intersection. A dimension $n$ complex projective space $\PP^n$ can be thought of as the closure of the affine space $\CC^n$ obtained by adding `points at infinity', the effect of this being that if we take the projective closure of two affine curves that don't intersect in affine space they will have intersections in projective space at the added points at infinity, for more details on affine and projective spaces see \cite{Mumford}. As a consequence of the discussion above the exact expressions we obtain for the ED degrees of the projective closures of our steady state varieties will be upper bounds on the the number of solutions of the associated affine ED problem. 

In what follows we will often wish to work with the projective closure of our affine steady state variety $V_\mathfrak{N}$. Let $k$ be a field (such as $\CC$ or $\RR$) and suppose $X$ is an affine variety in $k^n$, we will write $\overline{X}\subset \PP^n$ for the projective closure of $X$. Recall that projective varieties are defined by homogeneous polynomial equations. To obtain the projective closure of an affine variety we must \textit{homogenize} the equations of a Gr\"obner basis for the defining ideal. More explicitly, consider an affine variety $X=V(f_1,\dots,f_r)\subset k^n$ and suppose the polynomials $f_1,\dots, f_r$ form a Gr\"obner basis for an ideal $I=(f_1,\dots, f_r)$ in $k[x_1,\dots,x_n]$. We will homogenize the ideal $I$ by homogenizing all of the $f_i$ with respect to $x_0$ to obtain polynomials $f_i^h$ in $k[x_0,\dots,x_n]$. The projective closure $\overline{X}\subset \PP^n$ is then defined by the ideal $I^h=(f_1^h, \dots, f_r^h)\subset k[x_0,\dots,x_n]$, i.e.~$\overline{X}=V(f_1^h, \dots, f_r^h) \subset \PP^n$. 

Let $W_X\subset  \CC^{n+1}$ denote the affine cone over $\overline{X}$, hence, ${\rm EDdegree}(\overline{X})={\rm EDdegree}(W_X)$. Recall that a projective variety and its affine cone are defined by the same homogeneous ideal. All solutions to the optimization problem \eqref{eq:projectedOp} for an affine variety $X$ will have a corresponding solution in the affine cone over $X$, $W_X$, hence we have that $${\rm EDdegree}({X})\leq{\rm EDdegree}(W_X)={\rm EDdegree}(\overline{X}).$$

\subsection{Toric Models and Euclidean Distance Degree}\label{section:mathPrelim2}
In this subsection we briefly define projective toric varieties and summarize a combinatorial method to compute the Euclidean distance degree of a toric model. A more detailed discussion of this topic can be found in \cite{HS}.

A toric model is an algebraic variety defined as follows.
We fix an integer $d \times n$-matrix
 $A$, with columns $a_1,a_2,\ldots,a_n$, 
and rank $d$ such that the vector $(1,1,\ldots,1)$ lies in the row space of $A$ over $\QQ$. In particular, we allow $A$ to have negative entries.
Each column vector $a_i$ defines a (Laurent) monomial
$t^{a_i} = t_1^{a_{1i}} t_2^{a_{2i}} \cdots t_d^{a_{di}}$ where $t \in (\CC^*)^d$, and $\CC^* = \CC \backslash \{0\}$ denotes the non-zero complex numbers ($\CC^*$ is often referred to as the complex torus). 
The {\em affine toric variety} defined by $A$ is $$\,\tilde X_A\,=\,\overline{\{ (t^{a_1}, \ldots , t^{a_n}) \,:\,t \in (\CC^*)^d \}}\subset \CC^n,$$ that is, $\,\tilde X_A\,$ is the (Zariski) closure in $\CC^n$ of the image of the monomial parametrization specified by $A$. The implicit equations for $\,\tilde X_A\,$ will always be homogeneous binomials, that is $\,\tilde X_A\,=V(I)$ where I is an ideal defined by homogeneous binomial equations. More precisely \cite[Corrollary~4.3]{GBCP} tells us that \begin{equation}
I=\left( x^{c^+}-x^{c^-}\; | \; c \in \ker(A) \right) ,\label{eq:BinomialIdeal}
\end{equation}
 where $c^+_i$ is equal to $c_i$ if $c_i>0$ and $0$ otherwise, and where $c^-_i$ is equal to $|c_i|$ if $c_i<0$ and $0$ otherwise. Conversely, any prime polynomial ideal $(f_1,\dots,f_m)$ in $k[x_0,\dots,x_{n-1}]$, where each $f_i$ is a homogeneous binomial, will define an affine toric variety $\tilde{X}_A$ in $\CC^n$. Note that \eqref{eq:BinomialIdeal} gives us a simple way to transition between the parametric and implicit descriptions of a toric variety (and vice versa), namely by computing generators for the kernel of $A$ to obtain a list of vectors defining the implicit equations. Conversely we can compute the columns of $A$ from a prime binomial ideal since the exponents of the binomials, interpreted as in \eqref{eq:BinomialIdeal}, will define $\ker(A)$, from which $A$ may be computed. We note that this procedure may generate an isomorphic variety rather than equal variety, i.e.~we may change the embedding, however the ED degree, the degree of the variety, etc.~are invariant under isomorphism and do not depend on the embedding. 
 
The affine toric variety $ \tilde{X_A}$ is the affine cone over the {\em projective toric variety} $X_A \subset \PP^{n-1}$, that is $X_A$ is the (Zariski) closure in $\PP^{n-1}$ of the same parametrization. The projective toric variety $X_A$ is defined implicitly to be the zeros of the same set of homogeneous binomials that define $\tilde{X}_A$, that is $X_A=V(f_1,\dots,f_m)$ in $\PP^{n-1}$. We have that ${\rm dim}(X_A) = d-1$ and ${\rm dim}(\tilde X_A) = d$. To the projective toric variety $X_A$ we will associate a polytope $P={\rm Conv}(A)$, which is the convex hull of the lattice points specified by the columns of the matrix $A$. The polytope $P$ is contained in $\RR^d$ and has dimension $\dim(P)=\dim(X_A)=d-1$; the degree of $X_A$ may also be read from the polytope, namely $\deg(X_A)={\rm Vol}(P)$ where ${\rm Vol}$ denotes the normalized $d-1$ dimensional volume.
Further background on affine and projective toric varieties can be found in \cite{CLS,GBCP,GKZ,FultonToric}. 
{\begin{definition}[Dynamical System with Toric Steady States]
Consider a dynamical system $\dot{x}=f(x)$ where $f(x)$ is a system of polynomial equations in variables $x_1,\dots, x_n$. If the non-zero closure $(V(f))^{\neq 0}$ is a toric variety then we say that the dynamical system $\dot{x}=f(x)$ is a {dynamical system with toric steady states}.\label{def:DynSysWToircSteadyStates}
\end{definition}}
We now show that when studying the generic ED degree of a chemical reaction network with toric steady states the result is independent of the reaction rate constants. In all the follows we will consider only the ED problem as stated in \eqref{eq:opt3}.
\begin{proposition}
Let $X_A$ be the (toric part of the) steady state variety of a chemical reaction with toric steady states. The ED degree of $X_A$ is independent of the choice of reaction coefficients and is equal to the number of complex critical points of the unconstrained optimization problem: \begin{equation}
\label{eq:opt3}
 {\rm Minimize}\,\,  \sum_{i=1}^n {\lambda}_i ( {u}_i - t^{a_i}  )^2 \,\,\,
\hbox{over all $\,\,t = (t_1,\ldots,t_d) \in \RR^d$. } 
\end{equation} 
\end{proposition}
\begin{proof} Consider the ED problem for a chemical reaction network $\mathfrak{N}$ with toric steady states and a positive steady state $\tilde{x}=(\tilde{x}_1,\dots,\tilde{x}_n)\in \RR^n_{>0}$, which is a function of the rate constants $k$. Then the steady state variety $(V_{\mathfrak{N}})_A$ can be parameterized by a $d\times n$ matrix $A$ (with $(1,\dots,1)$ in its row space) as $$
(V_{\mathfrak{N}})_A=\,\overline{\{ (\tilde{x}_1t^{a_1}, \ldots , \tilde{x}_n t^{a_n}) \,:\,t \in (\CC^*)^d \}}\subset \CC^n.$$ The unconstrained version of the corresponding Euclidean distance minimization problem for $(V_{\mathfrak{N}})_A$ is:\begin{equation}
\label{eq:opt2}
 {\rm Minimize}\,\,  \sum_{i=1}^n \tilde{\lambda}_i ( \tilde{u}_i - \tilde{x}_it^{a_i}  )^2 \,\,\,
\hbox{over all $\,\,t = (t_1,\ldots,t_d) \in \RR^d$, } 
\end{equation}for generic $\tilde{u}$ and $\tilde{\lambda}$ (more precisely, by generic we mean for $\tilde{u}$ and $\tilde{\lambda}$ chosen from appropriate Zariski dense sets $D_{\tilde{u}}$ and $D_{\tilde{\lambda}}$). Observe that for any choice of $\tilde{x}=(\tilde{x}_1,\dots, \tilde{x}_n)\in (\CC^*)^n$ the sets $$D_{{u}}=\left\lbrace (u_1,\dots,u_n)\;|\; u_i=\frac{\tilde{u_i}}{\tilde{x}_i}\right\rbrace,\;\;\;D_{{\lambda}}=\left\lbrace (\lambda_1,\dots,\lambda_n)\;|\; \lambda_i=\tilde{x}_i{\tilde{\lambda_i}}\right\rbrace$$are also Zariski dense (since they are in one to one correspondence with a Zariski dense set). Hence we may state the ED problem as in \eqref{eq:opt3} for generic $\lambda\in D_{\lambda}$ and $u \in D_u$. From this it follows that the generic ED degree does not depend on the positive steady state $(\tilde{x_1},\dots, \tilde{x_n})$, nor on the reaction rates $k=(k_1,\dots, k_n)$, which determine $\tilde{x_1},\dots,\tilde{x_n}$. 
%In particular we may define the (generic) ED degree of a toric chemical reaction network as the number of complex critical points of the  minimization problem \eqref{eq:opt3}. 
\end{proof}

In \cite{HS} an exact formula for the ED degree for toric varieties was derived. Let $A$ be an integer $d\times n$ matrix with $(1,\dots,1)$ in its row space, parameterizing the toric component of the steady state variety, as above. It is shown in \cite{HS} that for the associated projective toric variety $X_A\subset \PP^{n-1}$ the Euclidean distance degree of $X_A$ can be computed combinatorially from the polytope $P={\rm Conv}(A)$. Specifically in \cite[Theorem~1.1]{HS} it is shown that \begin{equation}
{\rm EDdegree}(X_A) =\sum_{i=0}^{\dim(X_A)}(-1)^{d-i-1}\cdot \left(2^{d-1}-1\right)\cdot V_i\;, \label{eq:computeEDToric}
\end{equation}
where $V_i$ denotes the sum of all \textit{Chern-Mather} volumes of all dimension $i$ faces of the polytope $P$. When $P$ is a smooth polytope (so that $X_A$ is a smooth variety) $V_i$ is simply the sum of all the normalized $i$ dimensional volumes of all dimension $i$ faces of $P$. In the singular case the normalized volumes are weighted by the \textit{Euler obstruction} of a face, which is an integer that accounts for the singularities of $X_A$ associated to a face $\beta$ of the polytope $P$. The {\em Euler obstruction} of a face $\beta$ of $P$ is denoted ${\rm Eu}(\beta)$ and is defined recursively:
\begin{enumerate}
\item ${\rm Eu}(P)\,\,=\,\,1,$
\item $\displaystyle {\rm Eu}({\beta}) \quad =\sum_{{\alpha} {\rm \;s.t.}\; {\beta}{\;\rm is\; a} \atop                                                        
{\rm proper\; face \; of \;} {\alpha}}                                          
\!\!\!\!                                                                        
(-1)^{\dim({\alpha})-\dim({\beta})-1}\cdot { \mu(\alpha/\beta)} \cdot {\rm Eu}(\alpha).$
\end{enumerate}
Here {$ \mu(\alpha/\beta)$} is the normalized relative subdiagram volume (see \cite[Definition~2.1]{HS} or \cite[Definition~3.8]{GKZ}). Using this we define the {\em Chern-Mather volume} of a face  $\beta $ to be the product  of the normalized volume
${\rm Vol}(\beta)$ and the Euler obstruction ${\rm Eu}(\beta)$. With this notation the Chern-Mather of all dimension $i$ faces of $P$ is:
\begin{equation}
V_i \quad =\sum_{\beta {\rm \; face \; of\;} P \atop  \dim(\beta)=i}
\! {\rm Vol}(\beta){\rm Eu}(\beta),\label{eq:ED_CM_Formula}
\end{equation}see \cite[\S2]{HS} for more details. We now illustrate these definitions with an example. \begin{example}
Consider the $3\times 6$ integer matrix $$
A  =  \begin{pmatrix}
 1 & 0 & 1 & 2 & 3 & 1  \\
 0 & 1 & 1 & 1 & 1 & 2  \\
 1 & 1 & 1 & 1 & 1 & 1
 \end{pmatrix}.
$$The monomial parametrization of the associated projective toric variety $X_A\subset \PP^{6-1}$ is $$ X_A\,=\,\overline{\{ (t_1t_3: t_2t_3:t_1t_2t_3 : t_1^2t_2t_3:t_1^3t_2t_3:t_1t_2^2t_3) \,|\,t \in (\CC^*)^3 \}}\subset \PP^5.$$ Written as the common vanishing set of implicit homogeneous binomial equations $$
X_A=V({x}_{3}^{2}-{x}_{2} {x}_{4},{x}_{2} {x}_{3}-{x}_{1} {x}_{4},{x}_{1} {x}_{3}-{x}_{0} {x}_{5},{x}_{2}^{2}-{x}_{0} {x}_{5}).
$$From the matrix $A$ we see that $\dim(X_A)=3-1=2$. The polytope $P={\rm conv}(A)$ associated to $X_A$ is given in Figure \ref{fig:CMVoleExPolytope}.
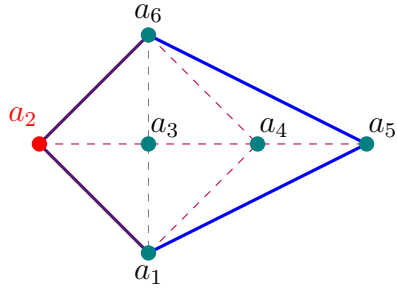
\begin{figure}
  \begin{center} \begin{tikzpicture}[scale=1.45]
\draw [purple,dashed](0,1) -- (3,1);
\draw [gray,dashed](1,2) -- (1,0);
\draw [purple,dashed](1,2) -- (2,1);
\draw [purple,dashed](1,0) -- (2,1);
\draw [pur1,very thick](0,1) -- (1,0);
\draw [pur1,very thick](0,1) -- (1,2);
\draw [blue,very thick](1,0) -- (3,1);
\draw [blue,very thick](1,2) -- (3,1);

%\node at (3,0) {$\color{mag1}P={\rm conv}(A)$};
\node at (1,-.2) {$a_1$};
\node at (-0.15,1.25) {$\color{red}a_2$};
\node at (1.15,1.15) {$a_3$};
\node at (2.15,1.15) {$a_4$};
\node at (3.15,1.15) {$a_5$};
\node at (1,2.2) {$a_6$};
\fill[red] (0,1) circle[radius=2pt];
\fill[teal] (3,1) circle[radius=2pt];
\fill[teal] (1,2) circle[radius=2pt];
\fill[teal] (1,0) circle[radius=2pt];
\fill[teal] (1,1) circle[radius=2pt];
\fill[teal] (2,1) circle[radius=2pt];
\end{tikzpicture}
\end{center}\caption{The polytope $P={\rm conv}(A)$. The columns of $A$ define the lattice points drawn above, note that we may draw this polytope $P$ in two dimensions since the third coordinate of all vectors is one, i.e.~they are all contained in the same plane in $\RR^3$.
\label{fig:CMVoleExPolytope}}
\end{figure}
From the triangulation in Figure \ref{fig:CMVoleExPolytope} we see that $\deg(X_A)={\rm Vol}(P)=6$. We now illustrate the computation of ${\rm EDdegree}(X_A)$ by computing the Chern-Mather volumes $V_0,V_1,V_2$ and applying \eqref{eq:computeEDToric}. Note that since the only dimension two face is the whole polytope $P$, and since by definition ${\rm Eu}(P)=1$, we have that $V_2={\rm Vol}(P)=\deg(X_A)=6$. Now compute $V_1$. One may check that $X_A$ has singularities only in dimension zero, hence all dimension one faces, i.e.~all edges ${\bf e}$ of $P$, will have Euler obstruction ${\rm Eu}({\bf e})=1$. From this we have that $V_1$ is equal to the number of edges of $P$, that is $V_1=4\cdot{\rm Vol}({\bf e}) =4$; it remains to compute $V_0$. 

Let ${\bf v}$ be a vertex contained in an edge ${\bf e}$ of the polytope $P$, to find ${\rm Eu}({\bf v})$ we must consider $ \mu({P}/{\bf v})$ and $ \mu({\bf e}/{\bf v})$. The subdiagram volume $ \mu({\bf e}/{\bf v})$ is equal to the one dimensional volume of an edge (which is one) minus the one dimensional volume of an edge with the vertex $v$ removed (which is zero), hence we always have that $ \mu({\bf e}/{\bf v})=1-0=1$. Consider the vertex ${\color{red}a_2}$ of $P$, the subdiagram volume $\mu(P/{\color{red}a_2})$ is equal to the two dimensional volume of $P$ minus the two dimensional volume of the convex hull of the points $a_1,a_3,a_4,a_5,a_6$ (that is the remaining points after ${\color{red}a_2}$ is removed). Hence $\mu(P/{\color{red}a_2})=\mathbf{6}-{\color{blue}\mathbf{4}}=2$, this gives $${\rm Eu}({\color{red}a_2})=2\cdot
 {\rm Eu}({\color{pur1}\mathbf{e}})\cdot\mu({\mathbf{e}}/{\color{red}a_2})
 -{\rm Eu}({P})\cdot\mu({P}/{\color{red}a_2})=2\cdot {\color{pur1}\mathbf{1}}-1\cdot 2
=0.$$ Similarly $ {\rm Eu}(a_5)=0$, and $ {\rm Eu}(a_1)={\rm Eu}(a_6)=-1,$ giving $$V_0=0 \cdot{\rm Vol}(a_2)+0 \cdot{\rm Vol}(a_5)-{\rm Vol}(a_1)-{\rm Vol}(a_6)=-2.$$Plugging these values into \eqref{eq:computeEDToric} we obtain
$$ {\rm EDdegree}(X_A) \,\, = \,\, 7 V_2 - 3 V_1 + V_0 \,\,=\,\,
7\cdot 6 - 3 \cdot 4 + (-2) \,=\, 28. $$

\end{example}
\subsection{Chemical Reaction Network Theory}
\label{subsection:CRNT_Intro}
In this subsection we briefly introduce chemical reaction networks and show how they give rise to dynamical systems \cite{Feinberg1987}. A more comprehensive overview of chemical reaction network theory can be found in \cite{Feinberg1987,Chellaboina2009,Gunawardena2003}. Consider a set of chemical {\em species} $\mathcal{S} = \{S_1,\dots,S_N\}$ and the vector of their respective concentrations $x = (x_1,\dots,x_N)$. A chemical reaction network is then defined as a weighted, directed graph $\mathcal{G} = (V,E,k)$. The vertex set $V$ consists of linear combinations of the chemical species,
\begin{equation}
C_i = \sum_{j=1}^N \alpha_{ij}S_j,
\label{eq:Comp}
\end{equation}
termed \emph{complexes}, such that $ C_i\in V$ and $i = \{1,\dots, M\}$
%= \mathcal{C} = \{C_1,\dots,C_{M}\}$
\cite{Gunawardena2003,Gunawardena2009}. The coefficients $\alpha_{ij} \in \ZZ_{\geq 0}$ are called \emph{stoichiometric coefficients} \cite{Chellaboina2009}. The edge set $E$ consists of the reactions, $C_i\rightarrow C_j$, with edge weights $k = \{k_1,\dots,k_L\}$.

The  reaction network can be embedded into $\RR^N$ by associating a standard basis vector of $\RR^N$, $e_i$, to each chemical species $S_i \in \mathcal{S}$. Arrange the reactions in any order and denote the $m^{th}$ reaction as $C_i \rightarrow C_j$. For each directed edge $e(C_i,C_j)\in E$ we define the reaction vector $r_m \in\RR^N$, as $r_m = \alpha_j - \alpha_i$. Here, $\alpha_i$ is the column vector of the stoichiometric coefficients of $C_i$ \cite{Gunawardena2003}.  We can convert the network description given by $\mathcal{G}$ to a system of ordinary differential equations (ODEs) by using the {\em law of mass action} which states that the reaction rates are proportional to the species concentrations \cite{Chellaboina2009}. Hence, we associate a monomial
\begin{equation}
x^{\alpha_i} = \prod_{j} x_j^{\alpha_{ij}},
\end{equation}
to each vertex $C_i \in V$ of $\mathcal{G}$. The directed edge $e(C_i,C_j)\in E$ has an edge weight of $k_m$, which provides the constant of proportionality for the law of mass action and gives rise to the $m^{th}$ element of the \emph{flux vector} \cite{Mahadevan2002}
\begin{equation}
R(x,k)_m = k_m x^{\alpha_i}.
\end{equation}
Similarly, we define the \emph{stoichiometric matrix} as
\begin{equation}
\Gamma = (r_1 \;r_2\;\cdots r_m\; \cdots \; r_L),
\end{equation}
with the set of reaction vectors $\{r_i\}$ as defined above.
The dynamics of the network can be described by the ODE system
\begin{equation}
\frac{dx}{dt} = \Gamma R(x,k).
\label{ODEs}
\end{equation}
It is apparent that the order of the reactions does not affect the system of equations \eqref{ODEs} as long as the elements of $R(x,k)$ and columns of $\Gamma$ are permuted equally.
In the remainder of this  paper we will define models by giving their flux vectors and stoichiometric matrices.

A link between dynamics and network structure is provided by {\em deficiency theory} \cite{Feinberg1987,Feinberg1988}. The {\em deficiency} of a chemical reaction network with mass action kinetics is given by
\begin{equation}
\delta = M-l- \text{dim}(\text{span}\{r_1,\dots,r_L\}),
\end{equation}
where $M = |V|$ is the number of complexes and $l$ is the number of connected components of $\mathcal{G}$.
It can be shown that certain classes of networks always have exactly one positive, stable steady state.
A {\em positive steady state} of a chemical reaction network is a concentration vector $x^* \in \RR^N_{>0}$ such that $\Gamma R(x^*) = 0$ \cite{ProcPhos}.
A reaction network is called \textit{weakly reversible} if, whenever there is a directed path in $\mathcal{G}$ from complex $C_i$ to $C_j$, then there also exists a directed path from $C_j$ to $C_i$ \cite{Feinberg1987}.
Deficiency theory provides one of the most important theorems in chemical reaction network theory, the {\em Deficiency Zero Theorem}.

\begin{theorem}[Deficiency Zero Theorem \cite{Feinberg1987}]
Consider a weakly reversible chemical reaction network with mass action kinetics. {
If such a network has deficiency zero, then the corresponding mass-action system has precisely one positive steady state for any choice of reaction rate parameters.}
The existence of the steady state is independent of the reaction parameters and the steady state is asymptotically stable.
\label{DefZeroThm}
\end{theorem}

{\begin{remark}
Networks which satisfy the Deficiency Zero Theorem are toric dynamical systems as studied in \cite{Craciun2009} and are sometimes also called complex balanced systems. The steady state varieties of these dynamical systems are toric \cite[Theorem~7]{Craciun2009}, hence these dynamical systems are dynamical systems with toric steady states in the sense of Definition \ref{def:DynSysWToircSteadyStates}. Since the toric dynamical systems of \cite{Craciun2009} are defined in terms of the reaction graph it is possible that a dynamical system with toric steady states (in the sense of Definition \ref{def:DynSysWToircSteadyStates}) is not a toric dynamical system (in the sense of \cite{Craciun2009}). For example, one can check that the two-site distributive phosphorylation model (see \S\ref{subsec:dist} or \cite[Example~3.13]{PerezMillan}) is not complex balanced, meaning the associated ODE system is not a toric dynamical system in the sense of \cite{Craciun2009}. On the other hand, the non-zero closure of the steady state variety is a toric variety, hence the associated ODE system is a dynamical system with toric steady states in the sense of Definition \ref{def:DynSysWToircSteadyStates}. In addition to satisfying the Deficiency Zero Theorem, another nice property of toric dynamical systems in the sense of \cite{Craciun2009} is that they may be identified from properties of a Euclidean embedding of the reaction graph, see \cite[\S2.1]{Yu2018}.
\end{remark}}

\begin{example}[Two-site kinetic proofreading]
A simple model of kinetic proofreading in T-cells follows the reaction scheme \cite{McKeithan1995}:
\begin{center}
\schemestart
A+B \arrow(1--2){<=>[$k_1$][$l_1$]} $X_1$ \arrow(@2--3){->[$k_2$]}[-90] $X_2$ \arrow(@3--@1){->[$l_2$]}
\schemestop
\end{center}
The dynamics of the network is governed by the ODE system
\begin{equation*}
\frac{d}{dt}
\begin{pmatrix}
a\\
b\\
x_1\\
x_2\\
\end{pmatrix} =
\underbrace{\begin{pmatrix}
-1 & 0 & 1& 1\\
-1 & 0 & 1& 1\\
1 & -1& -1& 0\\
0 &  1& 0& -1\\
\end{pmatrix}}_\Gamma
\underbrace{\begin{pmatrix}
k_1ab\\
k_2x_1\\
l_1x_1\\
l_2x_2\\
\end{pmatrix}}_{R(x)}.
%=
%\underbrace{\begin{pmatrix}
%1 & 0 & 0\\
%1 & 0 & 0\\
%0 & 1 & 0\\
%0 & 0 & 1\\
%\end{pmatrix}}_{Y^T}
%\underbrace{\begin{pmatrix}
%-k_1 & l_1 & l_2\\
%k_1 & -k_2-l_1 & 0\\
%0 & k_2 & -l_2
%\end{pmatrix}}_{\mathcal{L}_k^T}
%\underbrace{\begin{pmatrix}
%ab\\
%x_1\\
%x_2
%\end{pmatrix}}_{\Psi^Y(x)}.
\end{equation*}
It is easy to check the that the deficiency of this network is zero. The network is also clearly weakly reversible. Following the construction for the complex balancing ideal in \cite{Craciun2009} it can be shown that the toric component of the steady state variety is generated by the implicit equations
\begin{align*}
0 &= k_1ab - (l_1+k_2)x_1,\\
0 &= k_2x_1 - l_2x_2.
\end{align*}
\end{example}

\begin{example}[Michaelis-Menten Kinetics]
The standard Michaelis-Menten \cite{Michaelis1913} enzyme catalysis follows the reaction scheme
\begin{equation}
\text{E + S} \rightleftharpoons \text{ES} \rightarrow \text{E + P}.
\end{equation}
%The sets $\mathcal{S},\mathcal{C},\mathcal{R}$ have elements
%\begin{align*}
%\mathcal{S} &= \{\text{E, S, ES, P}\} = \{x_1,x_2,x_3,x_4\} = \left\{\begin{bmatrix}1\\0\\0\\0\end{bmatrix},\begin{bmatrix}0\\1\\0\\0\end{bmatrix},\begin{bmatrix}0\\0\\1\\0\end{bmatrix},\begin{bmatrix}0\\0\\0\\1\end{bmatrix}\right\},\\
%\mathcal{C} &= \{\text{E + S, ES, E+P}\} = \left\{\begin{bmatrix}1\\1\\0\\0\end{bmatrix},\begin{bmatrix}0\\0\\1\\0\end{bmatrix},\begin{bmatrix}1\\0\\0\\1\end{bmatrix}\right\},\\
%\mathcal{R} &= \{\text{E + S} \rightarrow \text{ES}, \text{ES} \rightarrow \text{E + S}, \text{ES} \rightarrow \text{E + P}\} = \left\{\begin{bmatrix}-1\\-1\\1\\0\end{bmatrix},\begin{bmatrix}1\\1\\-1\\0\end{bmatrix},\begin{bmatrix}1\\0\\-1\\1\end{bmatrix}\right\}.
%\end{align*}
Hence, we can formulate the governing ODE system
\begin{equation}
\frac{dx}{dt} = 
\underbrace{\begin{pmatrix}
-1 & 1 & 1\\
-1 & 1 & 0\\
1 & -1 & -1\\
0 & 0 & 1
\end{pmatrix}}_\Gamma
\underbrace{\begin{pmatrix}
k_1x_1x_2\\
k_2x_3\\
k_3x_3
\end{pmatrix}}_{R(x)}.
\end{equation}

It can be shown that $\delta = 0$, however, the system is not weakly reversible. Therefore, Michaelis-Menten kinetics does not have a steady state ideal with a toric irreducible component. In particular the steady state ideal is $I = \langle -k_1x_1x_2 + k_2x_3 + k_3x_3,\; -k_1x_1x_2+k_2x_3,\; k_1x_1x_2-k_2x_3-k_3x_3,\; k_3x_3 \rangle$, the primary decomposition of this ideal contains no components whose radical is a prime binomial ideal; hence the steady state variety cannot be written as a union of toric varieties.  

%which can be decomposed into prime ideals $I = I_1 \cup I_2 = \langle x_3,x_1 \rangle \cup \langle x_3,x_2\rangle$. Hence, the steady state ideal generates a variety which does not lie entirely in $(\CC^*)^3$ and therefore does not generate a toric variety.
\end{example}

\subsection{Real Closest Points}\label{subsection:realClosestPoints}
In the applications considered here we are particularly interested in the real solutions of the optimization problem \eqref{eq:projectedOp}. In \cite{FS1,FS2} it is shown that a finite, and non-zero, number of real solutions to \eqref{eq:projectedOp} always exist. We will briefly summarize the relevant results below.

For this discussion let $V=V(f_1,\dots,f_m)\subset \mathbb{R}^n$ be a real variety, we may suppose, without loss of generality, that $V$ is irreducible. For a general point $u\in \mathbb{R}^n$ and a general point $\lambda\in \mathbb{R}_{>0}^n$ we wish to solve the optimization problem \begin{equation}
{\rm Minimize}\;\;\; \sum_{i=1}^n \lambda_i (u_i-x_i)^2 \; \rm{subject \; to \;} x\in V .\label{eq:OpRe}
\end{equation}

The (positive weighted) square of the Euclidean norm is convex and differentiable. It follows by Theorem 2.2 of \cite{FS2} that for general $u\in \RR^n$ and general (positive) $\lambda\in \mathbb{R}_{>0}^n$ we have that there exists a unique real global minimum $x^*$ for the problem \eqref{eq:OpRe}. Further, since $x^*$ is a minimum it must be a solution to the critical equations associated to \eqref{eq:OpRe}. Note that the critical equations of \eqref{eq:OpRe} define a variety in $\RR^n$, hence if we consider the variety defined by these same equations in $\CC^n$ all points in the critical points variety in $\RR^n$ must appear. In particular the critical points variety in $\CC^n$ will contain the real global minimum $x^*$. It may also be shown (see Lemma 4.3 of \cite{FS1}) that all real minima are smooth critical points. Computing any global algorithmic solution to \eqref{eq:OpRe} in practice will require the computation of all complex critical points.

%Theorem 4.6 of \cite{FS1} and Theorem 5.2 of \cite{FS2} also give a more precise characterization of the real critical points variety. Lemma 4.3 of \cite{FS1} states that all real minima are smooth critical points. However computing any global algorithmic solution to \eqref{eq:OpRe} in practice will still require the computation of all complex critical points. 

For the case where $V=\tilde{X}_A$ is the affine cone over a projective toric variety $X_A$ we can make a statement regarding the existence of a \emph{positive} real local minimum. 
\begin{proposition}
Let $V=\tilde{X}_A$ be the affine cone over a projective toric variety $X_A$. Let $\mathbb{R}_+$ denote the positive reals and suppose that $V\cap {(\mathbb{R}_+)^n}$ is non-empty, then there exists a at least one real local minimum $x^*=(x_1^*,\dots, x_n^*)\in \RR^n$ for \eqref{eq:OpRe} such that $ x_i^*\geq 0, \; \forall i$.
\label{prop:real}
\end{proposition}\begin{proof}
%Following Fulton \cite[\S4.2]{FultonToric} we define an (algebraic) moment map $\mu:$

Fix a general real data point $u\in \RR^n$. The (positive weighted) squared Euclidean distance $$
f(x_1,\dots,x_n)= \sum_{i=1}^n \lambda_i  (x_i-u_i)^2, \;\;\; \lambda_i\in \RR_{>0}
$$is strictly convex in $\RR^n$. By \cite[\S4.2]{FultonToric} (see also \cite[\S12.2]{CLS}) there exists a map $\mu:X_A\to \RR^d$ (called the algebraic \textit{moment map}) which is a homeomorphism from the nonnegative part of $X_A$, $(X_A)_{\geq 0}$, onto the polytope $P={\rm conv}(A)$.  Further it is shown in \cite[\S4.2]{FultonToric} that the map $\mu$ induces a real analytic isomorphism between all points of $X_A\cap (\RR^*)^d$ and the relative interior of $P$, ${\rm int}(P)$. Similarly, the isomorphism induced by $\mu$ also gives an isomorphism between the torus orbit of a face and the relative interior of the face for any face of $P$. Since the set of real points in the polytope $P$ is a closed compact set it follows that $(X_A)_{\geq 0}$ is also a closed compact set with interior $(X_A)_{> 0}$. Let $\tilde{f}=f|_{(X_A)_{\geq 0}}$. The domain of $\tilde{f}$ is a closed compact set, therefore there exists at least one point $ x^*\in (X_A)_{\geq 0}$ that minimizes $\tilde{f}$. Since $x^*$ is a minimum of $\tilde{f}$ on $(X_A)_{\geq 0}$, and $f$ is defined by the same polynomial, then $x^*$ is both a critical point and a local minimum of $f$.    
\end{proof}
% }

\subsection{Computational Methods}\label{subsect:CompMeth}

%In order to illustrate how the measure of complexity given by the ED degree translates into the difficulty of selecting a model we solve the ED degree problem for our systems computationally and track the wall time of the calculations. We use a variety of standard symbolic and numerical techniques and this section serves as an outline of each with the associated benefits and weaknesses.
In this subsection we briefly review several standard methods for solving systems of polynomial equations. It should be emphasized that when studying the ED problem the complexity of the computational methods used corresponds primarily to the ED degree of the variety $X$ and not to $\deg(X)$. This is because when we solve the optimization problem \eqref{eq:projectedOp} we consider the zero dimensional variety consisting of the smooth points in $X$ which satisfy the critical equations of \eqref{eq:projectedOp}. While these points form a subvariety of $X$ the degree of this subvariety may be substantially different than that of $X$, and all computational methods will compute the solutions to \eqref{eq:projectedOp} by finding the points in this subvariety of $X$ defined by the critical equations.
\newline
\newline
\noindent\emph{Gr{\"o}bner bases:\;} Gr{\"o}bner basis methods have been shown to be very useful in the context of systems biology \cite{Gross2016b,Joshi2014}, especially due to the fact that they allow the user to find exact, symbolic expressions for the varieties concerned. Where possible, we compute the Gr{\"o}bner basis of our models in lexicographical (lex) monomial order to generate a triangular system. We can then find all real solutions of the system by iteratively applying a numerical or symbolic solver.  While highly useful, Gr{\"o}bner basis methods can be computationally expensive, particularly when computing a Gr{\"o}bner basis in the lex term ordering (to obtain a triangular system). More explicitly, effective methods to compute all points in a zero dimensional variety $W$ in a dimension $n$ ambient space using Gr{\"o}bner basis and often have two steps. First we find the Gr{\"o}bner basis in some other term order (which is faster to compute in) using some efficient Gr{\"o}bner basis algorithm. Second we apply the FGLM algorithm (or another reordering scheme) \cite{FGLM} to transform this into a lex Gr{\"o}bner basis. In practice the second step, namely the reordering step, is often the bottle neck. In the zero dimensional case the FGLM algorithm has complexity $\mathcal{O}(n \deg(W)^3)$ \cite{FGLM}, and in particular its complexity is primarily determined by the degree of $W$. In the case of the Euclidean distance problem, $\deg(W)$ is the degree of the variety defined by the critical equations of \eqref{eq:projectedOp}, i.e.~the ED degree. 
\newline
\newline
\noindent\emph{Numerical Algebraic Geometry (NAG):\;} NAG methods have recently been used for model selection and optimization problems \cite{Gross2016a}. We use two commonly applied NAG packages, PHCpack \cite{PHC} and Bertini \cite{Bertini}, to solve the ED problem for our chemical reaction networks. The advantage of parameter homotopy or NAG methods is that they can be much faster than current Gr{\"o}bner basis algorithms in some cases. However, realizing this benefit sometimes requires specially programed methods tailored to a given system. This can be aided by knowing, a priori, more details about the system to be solved, for example its degree. In our case knowing the ED degree ahead of time could be used in the construction of specialized NAG methods since this would represent the number of solution paths that would need to be tracked to solve the problem; hence this knowledge could improve both the performance and reliability of NAG methods in specific cases of interest. Implementations of NAG methods also provide black box solvers, however these may not be well suited for the particular problem at hand. In this paper we use PHCpack to compute the solutions to the ED problem for several examples. We also performed test computations using Bertini \cite{Bertini} however Bertini did not perform well on the toric varieties studied in this paper and we could not find solutions for any but the smallest ED degrees. Practical experience shows that PHCpack is often quite effective when applied to highly structured systems, such as those arising from toric varieties. However, as with all numeric methods, numerical stability and precision can pose challenges, making it hard to be certain all solutions have been computed, especially for larger ED degrees (see \S\ref{subsec:dist} for details).  %Therefore, if not all the solutions are present, there is no guarantee that all real (positive) minima will be in the computed solution set. This caveat can make NAG methods unreliable for solving the ED problem in cases where the ED degree is not known a priori. We note that a variety of certified NAG methods also could be applied, while these would also give all solutions, this process would be substantially slower than applying uncertified NAG methods and certifying using the ED degree in cases where the ED degree is known. 

%We aimed to contrast our calculations from PHCpack with the computations of another  widely used NAG package, Bertini. However, Bertini did not perform well on the toric varieties given in this paper and we could not find solutions for any but the smallest ED degrees, which due to the longer run time are omitted in this paper.
% normally faster than the Gr{\"o}bner basis method and taylored to toric varieties it has a tendency to miss solutions, 
\section{Results} %with Toric Steady States}
\label{section:models}
In this section we compute closed form formulas for the Euclidean distance degree of five commonly studied chemical reaction networks. For each model we also demonstrate how the ED degree helps us estimate the computational difficulty of the ED problem.
%the utility of knowing the ED degree ahead of time when performing model selection and model testing; in particular we see that this information aids in verifying the correctness of the solutions and

\subsection{Multi-site Phosphorylation Networks}

Phosphorylation is a ubiquitous mechanism in cell biology \cite{Salazar2009} and the most widely studied protein modification \cite{Cohen2001}. Phosphorylation controls the production of new proteins as well as their degradation and the transmission of intra- and intercellular signals. Abnormal phosphorylation is connected with a vast number of diseases such as cancer, diabetes, hypertension, heart attacks and rheumatoid arthritis \cite{Cohen2001}. In particular, abnormal myontonin phosphorylation leads to myotonic muscular dystrophy \cite{Roberts1997} and similarly, disturbed phosphorylation of the insulin receptor is a cause of diabetes \cite{Gual2005}.

A simple phosphorylation network consists of a substrate (a protein), kinases which phosphorylate the substrate and phosphatases to dephosphorylate. Phosphorylation can be thought of as a on/off switch for cellular mechanisms in which the presence of multiple sites enables fine tuning of such a switch \cite{Salazar2009}. For an $N$ site protein there exists, in principle, a maximum of $2^N$ states, which indicates a large redundancy in biological function. 

% The combinatorial complexity of a completely random multi-site phosphorylation network is large, 

While a protein with multiple sites can be phosphorylated in an arbitrary manner, there are two extreme mechanisms in multi-site phosphorylation. On the one hand, there is processive phosphorylation, where the kinase binds to the substrate and phosphorylates all sites before unbinding, and on the other hand we have distributive phosphorylation in which a binding-unbinding event is required for each phosphorylation. Experiments have highlighted the existence of kinases falling in each category, but also a whole spectrum of intermediate `processivity' \cite{Rubinstein}.

\subsubsection{Processive Networks}\label{subsec:proc}

Whilst purely processive systems are rarely found, there exist a number of cellular processes which exhibit a high degree of processivity \cite{Salazar2009,Velazquez2005,Aoki2011} such as the phosphorylation of the splicing factor  ASF/SF2 which has a role in heart development, cell motility and tissue formation \cite{Aubol2003,Ma2008,Burack1997}.

We consider the processive multisite phosphorylation network studied in \cite{ProcPhos} which, for $N$ sites, is described by the reaction scheme
\begin{align*}
&S_0 + E \xrightleftharpoons[k_2]{k_1} S_0E \xrightleftharpoons[k_4]{k_3} S_1E \xrightleftharpoons[k_6]{k_5}\cdots \xrightleftharpoons[k_{2N}]{k_{2N-1}} S_{N-1}E \xrightarrow{k_{2N+1}} S_N+E,\\
&S_N+F \xrightleftharpoons[l_{2N}]{l_{2N+1}} S_NF \xrightleftharpoons[l_{2N-2}]{l_{2N-1}}\cdots \xrightleftharpoons[l_4]{l_5} S_2F \xrightleftharpoons[l_2]{l_3} S_1F \xrightarrow{l_1} S_0+F.
\end{align*}

To translate this reaction network into a system of ordinary differential equations we assign variables representing the concentrations of the chemical species in the following way
\begin{center}
\begin{tabular}{ccccccccccc}
\hline
$x_1$ & $x_2$ & $x_3$ & $x_4$ & $x_5$ & $x_6$ & $x_7$ & $x_8$ & $\cdots$ & $x_{2N+3}$ & $x_{2N+4}$\\
\hline
$E$ & $F$ & $S_0$ & $S_N$ & $S_0E$ & $S_1F$ & $S_1E$ & $S_2F$ & $\cdots$ & $S_{N-1}E$ & $S_NF$\\
\hline
\end{tabular}
\end{center}
The assignment allows us to identify the variable $x_i$ with $e^i$, the $i^{th}$ basis vector of $\RR^{2N+4}$.
We can now formulate our system of ordinary differential equations via defining the matrix $\Gamma$ for the processive network
\begin{align}
\Gamma = &\left(e_5 - (e_1+e_3), \cdots, e_{2i +5} - e_{2i+3},\cdots,e_4+e_1-e_{2N+3},\right.\nonumber\\
&\left.e_2 + e_3 - e_6, \cdots, e_{2i+4}-e_{2i+6},\cdots, e_{2N+4} -(e_2+e_4)\right),
\label{GProc}
\end{align}
where $i = 1,\cdots,N-1$. Similarly the flux vector $R(x)$ can be formulated
\begin{equation}
R(x) =
\begin{pmatrix}
k_1x_1x_3 - k_2x_5\\
k_3x_5-k_4x_7\\
k_5x_7-k_6x_9\\
\vdots\\
k_{2N-1}x_{2N+1} - k_{2N}x_{2N+3}\\
k_{2N+1}x_{2N+3}\\
l_1x_6\\
l_3x_8 - l_2x_6\\
\vdots\\
l_{2N-3}x_{2N+2} - l_{2N-4}x_{2N}\\
l_{2N-1}x_{2N+4} - l_{2N-2}x_{2N+2}\\
l_{2N+1}x_2 x_4 - l_{2N}x_{2N+4}\\
\end{pmatrix}.
\label{RProc}
\end{equation}
From \eqref{ODEs} we have that the generators of the steady state ideal of the processive model are given by $\Gamma R(x)$. By considering the family of affine varieties generated by the steady state ideal of the processive model we can find the ED degree of the non-zero closure.

\begin{theorem}
Let $V_\mathfrak{N}$ denote the steady state variety of the $N$-site processive phosphorylation network. Then we have that $${\rm EDdegree}\left( \overline{V_\mathfrak{N}^{\neq 0}}\right)=28.$$  %Further, the A-discriminant $X^\vee$ is a hypersurface and $\deg(X)=\deg(X^{\vee})=4$. 
\end{theorem}\begin{proof}
In \cite[(5.11)]{ProcPhos} it is shown that $V_\mathfrak{N}^{\neq 0}=\lambda\cdot \tilde{X}_B$ where $\lambda\in (\CC^*)^{2N+2}$ and where $\tilde{X}_B$ is the affine toric variety defined by the $3 \times (2N+2) $ integer matrix $$
B=\begin{pmatrix}
0 & 1& 0& -1& 0&0 &\cdots & 0 \\
-1 &0& 1& 0& 0&0& \cdots &0\\
1 & 0& 0& 1&  1&1&\cdots & 1\\
\end{pmatrix}.
$$ Taking the projective closure of $V_\mathfrak{N}^{\neq 0}$ gives $\overline{V_\mathfrak{N}^{\neq 0}}=(1,\lambda)\cdot X_A$ where $X_A$ is the projective toric variety given by the $4 \times (2N+3) $ integer matrix $$
A=\begin{pmatrix}
1 & 1& 1& 0& 0&0&0 &\cdots & 0 \\
0 &-1& 0& 1& 0&0&0& \cdots &0\\
1 & 1& 0& 0& 1& 0&0&\cdots & 0\\
1 & 1& 1& 1& 1& 1&1&\cdots & 1\\
\end{pmatrix},
$$ where the last $2n-2$ columns are $(0,0,0,1)^{T}$. For any choice of the integer $n\geq 1$ we see that the matrix $A$ will have six unique column vectors, from this it is straightforward to see that the polytope $P={\rm conv}(A)$ has exactly six vertices with none of the columns of $A$ specifying interior lattice points. Since the Chern-Mather volumes are determined by the polytope $P$ and by the interior lattice points corresponding to columns of $A$ it follows immediately that we need only consider the six vertices of the polytope $P$. These are given as the columns of the matrix $$
P^{\rm vertex}=\begin{pmatrix}
1 & 1& 1& 0& 0&0 \\
0 &-1& 0& 1& 0&0\\
1 & 1& 0& 0& 1& 0\\
1 & 1& 1& 1& 1& 1\\
\end{pmatrix}=\begin{pmatrix} v_1 & v_2 &  v_3&  v_4 &  v_5 & v_6
\end{pmatrix},
$$ where $v_i$ denotes vertex $i$. %
%% \begin{figure}[h!]
%  \centering \myoct{2}
%  \caption{The polytope corresponding to the $A$ matrix of the (projective) steady state ideal of an $n$-site processive network.(this is a first draft, might make it a bit nicer later)}
% \end{figure}

Computing with this polytope we have that the Chern-Mather volumes of $P$ are $V_0=12$, $V_1=12$, $V_2=8$ and $V_3=4$. Applying \eqref{eq:ED_CM_Formula} gives $ {\rm EDdegree}\left(\overline{V_\mathfrak{N}^{\neq 0}}\right)=28$.

\end{proof}
%Applying Theorem 1.2 of \cite{HS} gives polar degrees $\delta_i=\delta_i(X_{A})$, $$
%(\delta_0,\delta_1, \delta_2,\delta_3)=(4,12, 8 , 4).
%$$
\subsubsection{Distributive Networks}\label{subsec:dist}

Distributive phosphorylation is an important mechanism in many cellular processes \cite{Salazar2009,Ferrell1997,Waas2001,Nash2001,Deschaies2001}, most prominently, the ERK2 MAP kinase, which is responsible for basic cellular functions such as cell proliferation, differentiation and cell death, is phosphorylated and dephosphorylated distributively \cite{Ferrell1997,Burack1997}.
In contrast to processive systems distributive multi-site phosphorylation networks can admit multiple steady states \cite{ProcPhos,Holstein}. Intuitively one would therefore expect the distributive mechanism to be `more complex' than its processive counterpart. In this section we quantify its complexity using the ED degree.

The reaction mechanism for the $N$-site distributive network from \cite{ProcPhos} is
\begin{align*}
&S_0 + E \xrightleftharpoons[l_1]{k_1} S_0E \xrightarrow{k'_1} S_1 + E \xrightleftharpoons[l_2]{k_2} S_1E \xrightarrow{k'_2} \cdots \xrightarrow{k'_{N-1}} S_{N-1} + E \xrightleftharpoons[l_N]{k_N} S_{N-1}E \xrightarrow{k'_N} S_N+E,\\
&S_N + F \xrightleftharpoons[\bar{l}_N]{\bar{k}_N} S_NF \xrightarrow{\bar{k}'_N} \cdots \xrightarrow{\bar{k}'_{N-2}} S_2+F \xrightleftharpoons[\bar{l}_2]{\bar{k}_2} S_2F \xrightarrow{\bar{k}'_2} S_1 + F \xrightleftharpoons[\bar{l}_1]{\bar{k}_1} S_1F \xrightarrow{\bar{k}'_1} S_0 + F.
\end{align*}

As in subsection \ref{subsec:proc} we identify the chemical species with the variables $\{x_1,\cdots,x_{3N+3}\}$:
\begin{center}
\begin{tabular}{ccccccccccc}
\hline
$x_1$ & $x_2$ & $x_3$ & $\cdots$ & $x_{N+3}$ & $x_{N+4}$ & $\cdots$ & $x_{2N+3}$ & $x_{2N+4}$ & $\cdots$ & $x_{3N+3}$\\
\hline
$E$ & $F$ & $S_0$ & $\cdots$ & $S_N$ & $S_0E$ & $\cdots$ & $S_{N-1}E$ & $S_1F$ & $\cdots$ & $S_NF$\\
\hline
\end{tabular}
\end{center}

From the reaction scheme we can derive the reaction matrix $\Gamma$ in terms of the matrices $\Gamma_1,\;\Gamma_2,\;\Gamma_3,\;\Gamma_4$,  the zero matrix $\bm{0}_{N\times N}$ and the identity matrix $\bm{I}_{N\times N}$, \scriptsize
$$
\Gamma_1 = \left(
 \begin{array}{ccc}
 -1 & \cdots & -1\\
 0 & \cdots & 0\\
 \hline
 & -\bm{I}_{N\times N} & \\
 \hline
 0 & \cdots & 0\\
 \hline
 & \bm{I}_{N\times N} & \\
 \hline
 & \bm{0}_{N\times N} & \\
 \end{array}\right), \;\; \Gamma_2 = \left(
 \begin{array}{ccc}
 1 & \cdots & 1\\
 0 & \cdots & 0\\
 \hline
 0 & \cdots & 0\\
 \hline
 & \bm{I}_{N\times N} & \\
 \hline
 & -\bm{I}_{N\times N} & \\
 \hline
 & \bm{0}_{N\times N} & \\
 \end{array}\right), \;\;\Gamma_3 = \left(
 \begin{array}{ccc}
 0 & \cdots & 0\\
 -1 & \cdots & -1\\
 \hline
 0 & \cdots & 0\\
 \hline
 & -\bm{I}_{N\times N} & \\
 \hline
 & \bm{0}_{N\times N} & \\
 \hline
 & \bm{I}_{N\times N} & \\
 \end{array}\right),\;\; \Gamma_4 = \left(
 \begin{array}{ccc}
 0 & \cdots & 0\\
 1 & \cdots & 1\\
 \hline
 & \bm{I}_{N\times N} & \\
 \hline
 0 & \cdots & 0\\
 \hline
 & \bm{0}_{N\times N} & \\
 \hline
 & -\bm{I}_{N\times N} & \\
 \end{array}\right),
$$
\normalsize
to give $\Gamma = \left(\Gamma_1\,|\,-\Gamma_1\,|\,\Gamma_2 \, | \, \Gamma_3 \, |\, -\Gamma_3\, |\, \Gamma_4\right)$. Analogously, the flux vector can be derived to give:
\begin{align}
R(x) =&
\left(
k_1x_1x_3,
\dots,
k_N x_1x_{N+2},
l_1x_{N+4},
\dots,
l_N x_{2N+3},
k'_1x_{N+4},
\dots,
k'_N x_{2N+3},\right.\nonumber\\
&\left.\bar{k}_1x_2x_4,
\dots,
\bar{k}_N x_2x_{N+3},
\bar{l}_1x_{2N+4},
\dots,
\bar{l}_N x_{3N+3},
\bar{k}'_1 x_{2N+4},
\dots,
\bar{k}'_N x_{3N+3}\right)^T.
\end{align}

We can now give closed form expressions for the ED degree of the family of toric varieties which corresponds to the non-zero closure of the steady state variety of the $N$-site distributive phosphorylation network. 

\begin{theorem}
Let $V_\mathfrak{N}$ denote the steady state variety of the distributive $N$-site phosphorylation network. Then we have that $${\rm EDdegree}\left( \overline{V_\mathfrak{N}^{\neq 0} }\right)=23N+5.$$  %Further, let $X_A=\overline{V_\mathfrak{N}^{\neq 0}}$ be the projective closure of $V_\mathfrak{N}^{\neq 0}$, the A-discriminant $X_A^\vee$ is a hypersurface, $\deg(X_{A})=3n+1$, and $\deg(X^{\vee}_{A})=4n$. 
\end{theorem}\begin{proof}
In Theorem 4.3 of \cite{CRSToric} it is shown that $V_\mathfrak{N}^{\neq 0}=\lambda\cdot \tilde{X}_B$ where $\lambda\in (\CC^*)^{3N+3}$ and where $\tilde{X}_B$ is the affine toric variety defined by $$
B=\begin{pmatrix}
0 &1 & 2& \cdots & N &1 & 2& \cdots & N &1 & 2 & \cdots & N & 1 & 0\\
0 &0 & 0& \cdots & 0 &1 & 1& \cdots & 1&1 & 1& \cdots & 1 &1 &1\\
1 &1 & 1& \cdots & 1 &1 & 1& \cdots & 1&1 & 1& \cdots & 1 &0 &0\\
\end{pmatrix}.
$$ Taking the projective closure of $V_\mathfrak{N}^{\neq 0}$ gives $\overline{V_\mathfrak{N}^{\neq 0}}=(1,\lambda)\cdot X_A$ where $(1,\lambda)\in (\CC^*)^{3N+4}$ and $X_A$ is the projective toric variety defined by the matrix $$
A=\begin{pmatrix}
0 &0 & 0&  1& 1 &\cdots & 1 &1 & 1& \cdots &1 & 1& 1& \cdots  &1\\
1 &0 & 0& 1& 1 &\cdots & 1 &0 & 0& \cdots &0 & 0& 0& \cdots  &0\\
0 &1 & 0& 0& 1 &\cdots & N &1 & 2& \cdots &N &1 & 2& \cdots &N\\
1 &0 & 1& N &N-1&\cdots & 0 &N-1 & N-2& \cdots &0 &N-1 & N-2& \cdots &0\\
\end{pmatrix}.
$$ $P_N={\rm conv}({A})$ is the polytope of the convex hull of the column vectors of $A$. Note that the matrix $A$ has exactly $7$ unique column vectors for any $N$. Given this is straightforward to see that for any $N$ the polytope $P_N$ always has the $7$ vertices given as the columns of the matrix $$
P_N^{\rm vertex}=\begin{pmatrix}
0 &0 & 0& 1& 1 & 1 &1\\
1 &0 & 0& 1 & 1 &0 &0\\
0 &1 & 0& N & 0 &N &1\\
1 &0 & 1& 0 & N &0&N-1 \\
\end{pmatrix}=\begin{pmatrix} v_1 & v_2 &  v_3&  v_4 &  v_5 &v_6&v_7
\end{pmatrix},
$$ where $v_i$ denotes vertex $i$. We now compute the Chern-Mather volumes of the faces of the polytope $P_N$ by applying \eqref{eq:ED_CM_Formula}. In dimension three we calculate that the Chern-Mather volume of $P_N$ is $V_3={\rm Vol}(P_N)=3N+1$. $P_N$ has 8 dimension two faces; totalling the Chern-Mather volumes of these faces gives $V_2=4\cdot(N+1)$. $P_N$ has 13 dimension one faces; totalling the Chern-Mather volumes of these faces gives $V_1=2N+10$. Finally totalling the Chern-Mather volumes of the vertices gives $V_0=12$. By \eqref{eq:computeEDToric} we have that \small $$
{\rm EDdegree}\left( \overline{V_\mathfrak{N}^{\neq 0} }\right)=15V_3-7V_2+3V_1-V_0=45N+15-(28N+28)+(6N+30)-12=23N+5.
$$%Further, applying Theorem 1.2 of \cite{HS}, we have that the polar degrees $\delta_i=\delta_i(X_{A})$ are $$
%(\delta_0,\delta_1, \delta_2,\delta_3)=(4n,8n+4,8n,3n+1).
%$$
\normalsize
\end{proof}

\begin{figure}[h!]
\centering
\resizebox{0.8\textwidth}{!}{%
\begin{minipage}{0.5\textwidth}
\includegraphics[width=\textwidth]{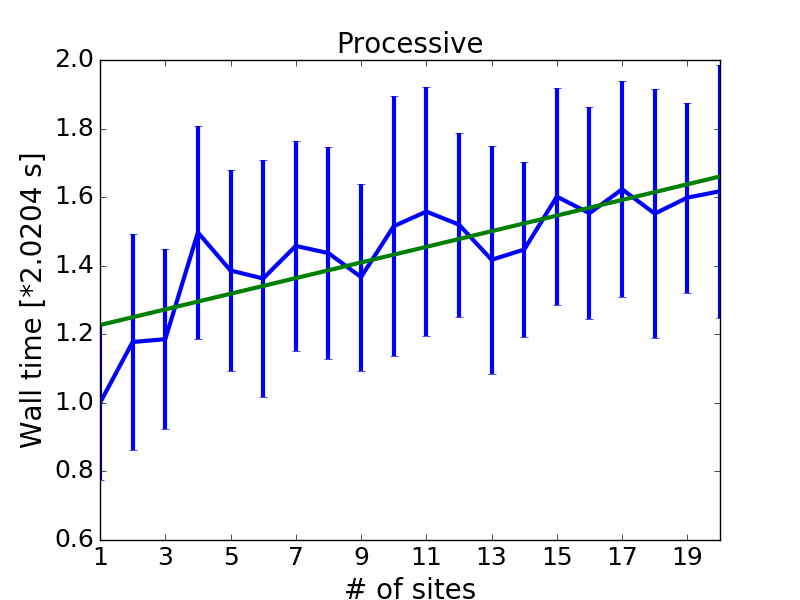}
\subcaption{Processive}
\label{fig:Proc}
\end{minipage}%
\begin{minipage}{0.5\textwidth}
\includegraphics[width=\textwidth]{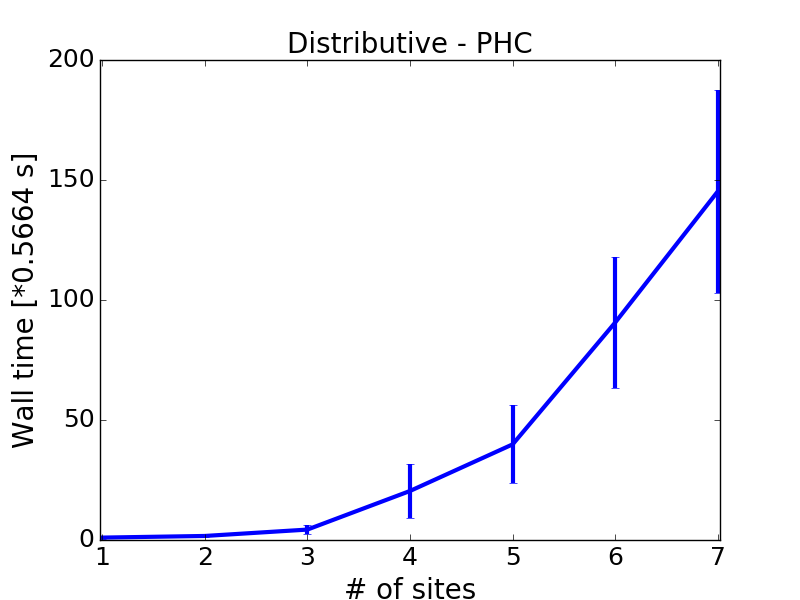}
\subcaption{Distributive}
\end{minipage}}
\caption{Comparing the run times to solve the ED problem for the processive and the distributive multi-site phosphorylation networks. For the processive model we computed the Gr{\"o}bner basis whereas for the distributive we use NAG methods, because the run time was considerably shorter than for computing the Gr{\"o}bner basis.}
\label{fig:Phos}
\end{figure}
% Since the Gr{\"o}bner basis is triangular we can apply a standard numerical solver to find all the real solutions and hence the global minimum.
For small $N$ we can solve the ED problem for the multi-site phosphorylation networks computationally in SageMath \cite{SageMath} by computing the Gr{\"o}bner basis of the non-zero closures of the models' steady varieties in lex order. In Figure \eqref{fig:Proc} we show that the the wall time for computations of the Gr{\"o}bner basis is approximately constant for the processive network. The run times for the computation of the Gr{\"o}bner bases of the distributive network prove very large ($>$1 day for N=2 sites), hence, for efficiency, we use PHCpack \cite{PHC}. Despite the NAG method being orders of magnitude faster we encounter the problem of PHCpack missing solutions, especially as the number of sites increases. Since there is no guarantee that the global minimum is in the solutions that are found, many runs are needed to find all solutions to the ED problem (see Table \eqref{tab:Dist}). Thus, in particular, the ED degree aids in determining whether all solutions have been found. This, combined with a method to certify numerical solutions, can give a certain answer to when all solutions have been found using NAG methods. %Alternatively, the more widely known NAG software Bertini could be used, but for toric varieties it is orders of magnitudes slower than PHCpack.

%\begin{table}[h!]
%\centering
%\resizebox{.95\linewidth}{!}{
%\begin{tabular}{@{} *6c @{}}
%\toprule 
% \multicolumn{1}{c}     {\color{Ftitle} Number of Sites}  &   {\color{Ftitle} Avg.~\% of runs giving correct answer} &   {\color{Ftitle} Avg.~time per run} &   {\color{Ftitle} Avg.~total run time per correct answer}  \\ 
% \midrule 
%  \color{line}1 & \color{line}97\% & \color{line}0.30s& \color{line}0.31s \\ 
% \color{line}2 & \color{line}96\% & \color{line}0.86s& \color{line}0.90s \\ 
%  \color{line}3 & \color{line}49\% & \color{line}1.92s& \color{line}5.99s \\ 
 %   \color{line}4 & \color{line}17\% & \color{line}2.94s& \color{line}25.31s \\ 
%\bottomrule
% \end{tabular}}\vspace{1mm}
%\caption{Run time results for the Distributive network using PHCpack} \label{tab:Dist}
% \end{table}
\begin{table}[h!]
\centering
\resizebox{.95\linewidth}{!}{
\begin{tabular}{@{} *6c @{}}
\toprule 
 \multicolumn{1}{c}     {\color{Ftitle} Number of Sites}  &   {\color{Ftitle} Avg.~\# of runs for solution} &   {\color{Ftitle} Avg.~time per run} &   {\color{Ftitle} Avg.~total run time per solution}  \\ 
 \midrule 
  \color{line}1 & \color{line}1.2 & \color{line}0.47s& \color{line}0.57s \\ 
 \color{line}2 & \color{line}1.1 & \color{line}0.90s& \color{line}0.95s \\ 
  \color{line}3 & \color{line}1.3 & \color{line}1.90s& \color{line}2.46s \\ 
    \color{line}4 & \color{line}2.6 & \color{line}4.41s& \color{line}11.57s \\ 
        \color{line}5 & \color{line}4 & \color{line}5.75s& \color{line}23.00s \\ 
        \color{line}6 & \color{line}4.6 & \color{line}11.24s& \color{line}51.36s \\ 
        \color{line}7 & \color{line}6.17 & \color{line}13.33s& \color{line}82.20s \\ 
\bottomrule
 \end{tabular}}\vspace{1mm}
\caption{Run time results for the Distributive network using PHCpack, the total time to find all solutions is listed in the last column. } \label{tab:Dist}
 \end{table}

\subsection{Sequestration Networks}\label{subsec:seq}

Sequestration reactions are chemical reactions in which a molecule is rendered inactive by binding to a second molecule \cite{Joshi2014}. A classic example would be the inhibition of a substrate by an enzyme,
\begin{equation}
E+S \rightarrow \emptyset.
\end{equation}

Following \cite{Joshi2014} a sequestration network of $N$ species is constructed by allowing for $N-1$ sequestration reactions and one synthesis reaction. We will let $\mathfrak{N}_{m}$ denote the $N$ species sequestration network and let $X_i$ denote the species occurring in $\mathfrak{N}_{m}$, the reaction scheme is given below:
\begin{align}
X_1 + X_2 &\xrightarrow{k_1} \emptyset,\nonumber\\
X_2 + X_3 &\xrightarrow{k_2} \emptyset,\nonumber\\
&\vdots\nonumber\\
X_{N-1} + X_{N} &\xrightarrow{k_{N-1}} \emptyset,\nonumber\\
X_1 &\xrightarrow{k_N} mX_N.\label{eq:SeqReactScheme}
\end{align}%Hence, we can define the species vector $\bm{x}$ and the flux vector $\bm{\gamma}$ as
Denoting the concentration of $X_i$ as $x_i$ the flux vector is
\begin{equation}
R(x) =
\left(
k_1x_1x_2,
k_2x_2x_3,
\dots,
k_{N-1}x_{N-1}x_N,
k_Nx_1
\right)^T. \label{eq:Rx_Sequest}
\end{equation}
From the reaction scheme we can find the matrix $\Gamma$ which is given by
\begin{equation}
\bm{\Gamma} = (-e_1-e_2, -e_2-e_3,\cdots,-e_{N-1}-e_N, -e_1+ m\;e_N).\label{eq:Gamma_Sequest}
\end{equation}

After closer inspection one finds that not every $\mathfrak{N}_{m}$ has a non-empty non-zero closure or a toric steady state variety. In the proof of Theorem \ref{theorem:Seq} we show the toric steady states only for certain choices of $N$ and $m$ and we calculate the ED degree for the case where $m=1$ and $N$ is an odd integer.

\begin{theorem}
Let $V_{\mathfrak{N}_{m}}$ be the steady state variety of the N-site sequestration network $\mathfrak{N}_{m}$ with $m$ being the coefficient of the synthesis reaction in the reaction network as in \eqref{eq:SeqReactScheme}.
If $m=1$ and $N$ is an odd integer then we have that $${\rm EDdegree}\left( \overline{V_{\mathfrak{N}_{m}}^{\neq 0}}\right)=1.$$ % Further ${\rm codim}(X^{\vee}_{A})=3$ and $\deg(X_{A_m})=\deg(X^{\vee}_{A_m})=1$.
For all other choices of $m$ and $N$ the variety $V_{\mathfrak{N}_{m}}^{\neq 0}$ is empty.
\label{theorem:Seq}
\end{theorem}\begin{proof}
First let $N$ be an odd integer and treat $m$ as a variable. Examining the matrix $\Gamma$ in \eqref{eq:Gamma_Sequest} we see that each row of $\Gamma$ has only two non-zero entries. Also note that each entry of the vector $R(x)$ in \eqref{eq:Rx_Sequest} is  a monomial, hence the ideal $I_\Gamma=\Gamma\cdot R(x)$ is generated by binomials. By definition the steady state variety is $V_{\mathfrak{N}_{m}}=V(I_\Gamma)$, the variety $V_{\mathfrak{N}_{m}}^{\neq 0}$ consists of all points in $V(I_\Gamma)$ which have no zero coordinates, hence in particular we have $$
V_{\mathfrak{N}_{m}}^{\neq 0}=\overline{V(I_\Gamma) \backslash V(x_1x_2\cdots x_N) }=V(I_\Gamma:(x_1x_2\cdots x_N)^\infty).
$$ Computing the ideal $I^{\neq 0}=I_\Gamma:(x_1x_2\cdots x_N)^\infty$ we obtain a prime ideal generated by binomials, one of these binomials is $m-1$. Hence $V_{\mathfrak{N}_{m}}^{\neq 0}$ is empty whenever $m\neq 1$. Setting $m=1$ in $I^{\neq 0}$ we obtain a new ideal $I_{m=1}^{\neq 0}$, again generated by binomials. Interpreting the exponents of these binomials as in \eqref{eq:BinomialIdeal} we obtain the generators of the kernel of the matrix which defines a parametrization of a variety isomorphic to $V_{\mathfrak{N}_{m=1}}^{\neq 0}=V(I_{m=1}^{\neq 0})$. We call the resulting toric variety $X_{B}$; the defining matrix of $X_{B}$ is the $2 \times (N+1)$ integer matrix 
$$
B=\begin{pmatrix}
1 & 0 & 0 &0 & 0 & \cdots & 0 & 0 & 0 \\
0 & 1 & 0 &1 & 0 & \cdots & 1 & 0 & 1  
\end{pmatrix}.
$$

Now suppose that $N$ is an even integer and again treat $m$ as a variable. Computing $I^{\neq 0}=I_\Gamma:(x_1x_2\cdots x_N)^\infty$ in this case again yields a prime ideal generated by binomials, among these binomials is the polynomial $m+1$. However, $m=-1$ is not valid for our model ($m$ must be positive by construction), hence $V_{\mathfrak{N}_{m}}^{\neq 0}$ is empty in this case. 

We now consider the case where $V_{\mathfrak{N}_{m}}^{\neq 0}$ is non-empty, namely we set $m=1$ and let $N$ be an odd integer. By the arguments above the non-zero closure of $V_{\mathfrak{N}_{m}}$ is given by $V_{\mathfrak{N}_{m}}^{\neq 0}=X_{B}$. 
Taking the projective closure of $X_{B}$ yields the projective toric variety $X_A$ defined by the $3 \times (N+2) $ integer matrix $$
A=\begin{pmatrix}
1  & 0 & 0 & 0 & \cdots & 0 & 0 & 0\\
0  & 1 & 0 & 1 & \cdots & 0 & 1 & 0\\
1 & 1 & 1 & 1 & \cdots & 1 & 1 & 1\\
\end{pmatrix}.
$$ The polytope $P={\rm conv}(A)$ is a triangle of dimension $2$ (see Figure \ref{fig:Seq}). The three vertices of $P$ are given as the columns of the matrix $$
P_m^{\rm vertex}=\begin{pmatrix}
1 & 0& 0 \\
0 & 1& 0\\
1 & 1& 1\\
\end{pmatrix}=\begin{pmatrix} v_1 & v_2 &  v_3
\end{pmatrix},
$$ where $v_i$ denotes vertex $i$. \begin{figure}[h!]
 \centering
   %  \includegraphics[width=0.3\textwidth]{Polytope1}
    % \vspace{-0.15in}
    \begin{tikzpicture}[scale=1.4]
\draw [blue,very thick](0,0) -- (0,1);
\draw [blue,very thick](0,1) -- (1,0);
\draw [blue,very thick](0,0) -- (1,0);

\node at (1,-.2) {$v_1$};
\node at (-0.15,1.15) {$v_2$};
\node at (0,-0.2) {$v_3$};
\fill[pur1] (0,1) circle[radius=2pt];
\fill[pur1] (1,0) circle[radius=2pt];
\fill[pur1] (0,0) circle[radius=2pt];
\end{tikzpicture}
     \caption{The polygon $P  = {\rm conv}(A)$. \label{fig:Seq}
     }
\end{figure}
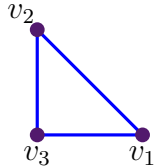

This polytope $P$ (see Figure \ref{fig:Seq}) is smooth (i.e.~the associated toric variety is smooth), hence $V_j$, the sum of all Chern-Mather volumes of all faces of dimension $j$, is equal to the sum of all normalized volumes of dimension $j$ faces. $P$ is a triangle, and hence has normalized volume one, has three edges (dimension one faces) each of with normalized volume one, and has three vertices. This gives $V_2=1$, $V_1=3$, $V_0=3$, respectively. Applying \eqref{eq:computeEDToric} gives $$ {\rm EDdegree}(X_{A})=\sum_{j=0}^2(-1)^{2-j}(2^{j+1}-1)\cdot V_j=1.$$
\end{proof}
%Applying Theorem 1.2 of \cite{HS} gives polar degrees $\delta_i=\delta_i(X_{A})$, $$
%(\delta_0,\delta_1, \delta_2)=(0,0,1).$$
\begin{remark}
Theorem \ref{theorem:Seq} is rather remarkable as it states that for any parameter vector and any measurement there will always be exactly one local minimum for the Euclidean distance problem associated to the sequestration network \eqref{eq:SeqReactScheme}. However, it can be shown that $V_{\mathfrak{N}_{m=1}}^{\neq 0}$ does not intersect the positive orthant and, therefore, while sequestration networks are biologically important, the study of their positive steady states is futile.
% Theorem \ref{theorem:Seq} is rather remarkable as it tells us that for any parameter vector and any measurement we will always find exactly one local minimum for the Euclidean distance problem associated to the sequestration network \eqref{eq:SeqReactScheme}. Using Proposition \ref{prop:real} we can guarantee that all coordinates of this minimum point are positive real numbers; in particular this minimum is the desired model point.
% Such networks for which every measured point corresponds to exactly one real point on the model are called identifiable \cite{Craciun2008}.
\end{remark}

\begin{figure}
\centering
\resizebox{0.8\textwidth}{!}{%
\begin{minipage}{0.5\textwidth}
\includegraphics[width=\textwidth]{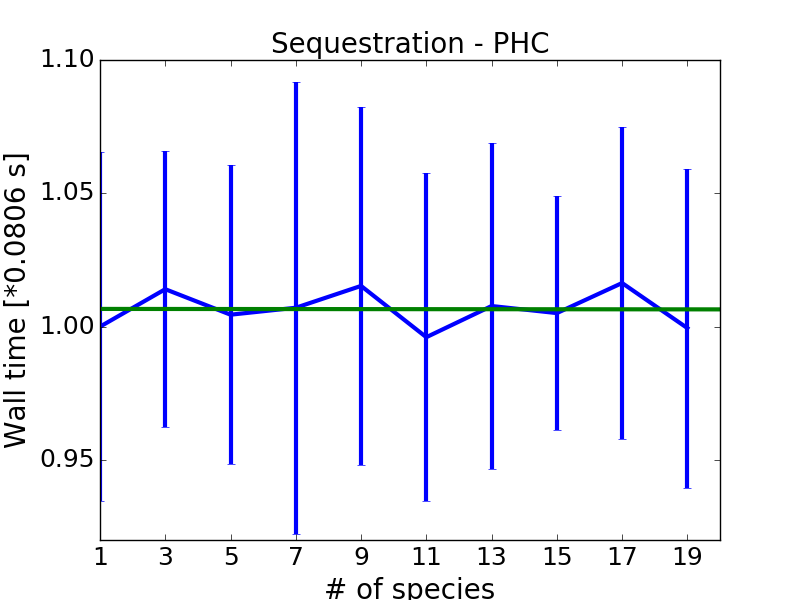}
\subcaption{Direct computation with PHCpack}
\end{minipage}%
\begin{minipage}{0.5\textwidth}
\includegraphics[width=\textwidth]{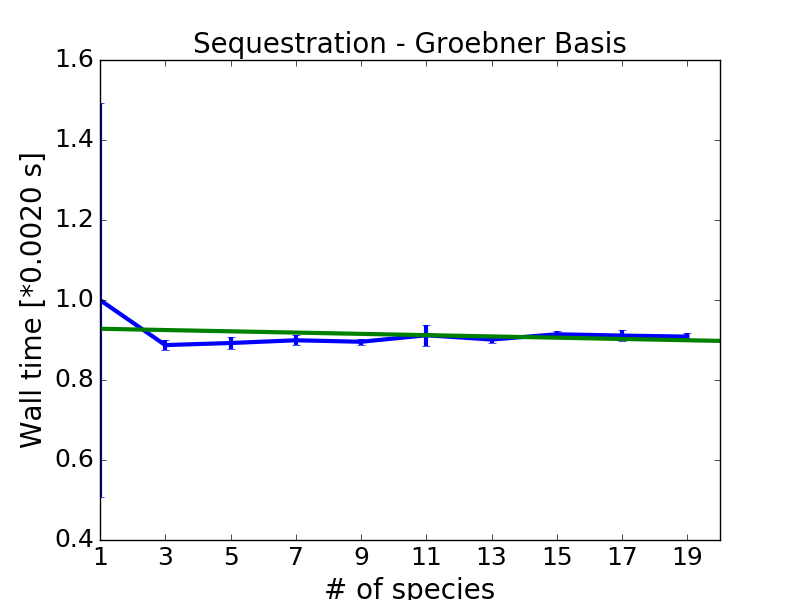}
\subcaption{Gr{\"o}bner basis only}
\end{minipage}}
\caption{A comparison of the run times for finding the solution to the ED problem for the Sequestration networks \eqref{eq:SeqReactScheme}. NAG and Gr{\"o}bner basis methods give an approximately constant run time for any number of species as expected from the constant ED degree.}
\label{fig:SeqTim}
\end{figure}

Due to its minimal complexity the ED problem for sequestration networks can be solved by computing the Gr{\"o}bner basis from which the solution can be read off immediately; see Figure \ref{fig:SeqTim}. 
% the solution directly with PHCpack, however, a false, complex solution will sometimes be found. The risk of a dummy solution can be eliminated by computing .

\subsection{Kinetic Proofreading Networks (McKeithan Model)}\label{subsec:mck}

Kinetic proofreading networks are vital components in cell biology that enhance binding selectivity \cite{Blanchard2004,McKeithan1995,Hopfield1974,Sontag2001}. The first mathematical models for kinetic proofreading were developed to explain the astonishing accuracy of DNA replication and protein synthesis. The model we present in this paper was initially proposed by McKeithan \cite{McKeithan1995} to understand the simultaneous high sensitivity and high selectivity of antigen recognition in T cells.

In the model a ligand ($A$) binds to T-cell receptor ($B$) which is transformed via intermediate stages to a final stage $X_N$. The product $X_N$ initiates the T-cell reaction. There is time-delay between the initial binding and the immune reaction. Therefore, ligands which are not tailored to a specific T-cell receptor will dissociate before $X_N$ is formed in significant quantities.
The \textit{McKeithan model} follows the reaction scheme:

%\begin{figure}[h!]
\begin{center}
\begin{tikzpicture}[->,>=stealth',shorten >=1pt,auto,node distance=1.7cm, semithick,scale=0.5]
  \tikzstyle{every state}=[fill=none,draw=none,text=black]
  \node[state] (0) {$A+B$};
  \node[state] (1)[right of = 0] {$X_1$};
  \node[state] (2)[right of = 1] {$X_2$};
  \node[state] (3)[right of = 2] {$\cdots$};
  \node[state] (N)[right of = 3] {$X_N$};
  
  \path (0) edge node {$k_1$} (1)
  		(1)	edge node {$k_2$} (2)
  		(2) edge node {$k_3$} (3)
  		(3) edge node {$k_N$} (N)
  		(N) edge[loop left, in = 90, out=110] node[above] {$l_N$} (0)
  		(2) edge[loop left, looseness=1.2, in = 80, out=110] node[above] {$l_2$} (0)
  		(1) edge[loop left, in = 70, out=110] node[above] {$l_1$} (0);
 % \path (A) edge[loop left, looseness=15, in = 155,out=-155, min distance=20] node {} (A)
 % 		(B) edge[loop right, looseness=15, in = 25,out=-25, min distance=20] node {} (B);
\end{tikzpicture}
\end{center}
%\caption{The reaction scheme of the McKeithan model for kinetic proofreading}
%\label{McKModel}
%\end{figure}
%The set of species is $\mathcal{S} = \{A,B,X_1,\cdots,X_N\}$ and the set of complexes is $\mathcal{C} = \{A+B, X_1,\cdots,X_N\}$ with $|\mathcal{C}| = N+1$.
%By letting $\mathcal{B}=\{e_i: i\in \{1,\cdots,N+2\}\}$ be the standard basis for $\mathbb{R}^{N+2}$ we can write the reaction set as
%$$
%\mathcal{R} = \{-e_1-e_2+e_3,-e_3+e_4,\cdots,-e_{N+1}+e_{N+2},e_1+e_2-e_3,\cdots,e_1+e_2,-e_{N+2}\},
%$$ 
%and notice that $\text{dim}(\text{span}\;\mathcal{R}) = N$.
%Hence, the model is a deficiency zero system with toric steady states.
Denoting the concentrations of the species as $\{a,b,x_1,\cdots,x_N\}$ we can formulate the flux vector
%Translating the network into a system of coupled ODEs we can define the species vector $\bm{x}$ and the flux vector $\bm{\gamma}$ as
\begin{equation}
R(x) =
\left(
k_1ab,
k_2x_1,
k_3x_2,
\dots,
k_Nx_{N-1},
l_1x_1,
l_2x_2,
\dots,
l_Nx_N
\right)^T.
\end{equation}
As in the previous models the stoichiometric matrix can be constructed from the reaction network
\begin{equation}
\bm{\Gamma} = (-e_1-e_2+e_3,-e_3+e_4,\cdots,-e_{N+1}+e_{N+2},e_1+e_2-e_3,\cdots,e_1+e_2-e_{N+2}).
\end{equation}

By using \eqref{ODEs} we find the affine steady state variety of the model and its parametrization.

% \begin{lemma}
% The steady state variety of the n-site McKeithan Model is the affine toric variety $X_{B_n}$ defined by the $2 \times (n+2)$ integer matrix $$
% B_n=\begin{pmatrix}
% 1 & 1 & 1 & \cdots & 1 & 1 & 0 \\
% 1 & 1 & 1 & \cdots & 1 & 0 & 1  
% \end{pmatrix}.
% $$ Further $\dim(X_{B_n})=2$.
% \end{lemma}
% \begin{proof}
% By induction and theorems from \cite{ProcPhos,CRSToric}.
% \end{proof}

We are now in a position to calculate the projective closure and hence the ED degree.

\begin{theorem}
Let $V_\mathfrak{N}$ be the steady state variety of the $N$-site McKeithan Model. Then we have that $${\rm EDdegree}\left( \overline{V_\mathfrak{N}^{\neq 0} }\right)=6.$$    
%Further the A-discriminant $X_A^\vee$ is a hypersurface and $\deg(X_{A_m})=\deg(X^{\vee}_{A_m})=2$. 
\end{theorem}\begin{proof}
Let $I_{\Gamma}$ be the ideal in $k[x_1,\dots,x_N,a,b]$ generated by $\Gamma R(x) $. Computing a Gr\"obner basis of $I_\Gamma$ (in the graded reverse lexicographic order) we find that $I_{\Gamma}$ is a prime ideal generated by binomials (for $N\geq 1$). Finding a matrix with kernel given by the set of exponents of $I_\Gamma$ (interpreted as in \eqref{eq:BinomialIdeal}) we have that the steady state variety of the $N$-site McKeithan Model is isomorphic to the affine toric variety $X_{B}$ defined by the $2 \times (N+2)$ integer matrix $$
B=\begin{pmatrix}
1 & 1 & 1 & \cdots & 1 & 1 & 0 \\
1 & 1 & 1 & \cdots & 1 & 0 & 1  
\end{pmatrix}.
$$ Since any toric variety is its own non-zero closure we have that $ V_\mathfrak{N}^{\neq 0}\cong X_{B}$. Taking the projective closure of $X_{B}$ we find that $\overline{X_{B}} = X_A$ is the projective toric variety defined by the $3 \times (N+3) $ integer matrix $$
A=\begin{pmatrix}
1  & 0 & 0 & 0 & \cdots & 0 & 1 & 0\\
1  & 0 & 0 & 0 & \cdots & 0 & 0 & 1\\
1 & 1 & 1 & 1 & \cdots & 1 & 1 & 1\\
\end{pmatrix}.
$$ For any choice of the integer $N\geq 1$ we have that the polytope $P={\rm conv}(A)$ has exactly four vertices with none of the columns of $A$ specifying interior lattice points. Since the Chern-Mather volumes are determined by the polytope $P$ and by the interior lattice points corresponding to columns of $A$ it follows immediately that we need only consider the four vertices of the polytope $P$. These are given as the columns of the matrix $$
P^{\rm vertex}=\begin{pmatrix}
1 & 0& 1& 0 \\
1 & 0& 0& 1\\
1 & 1& 1& 1\\
\end{pmatrix}=\begin{pmatrix} v_1 & v_2 &  v_3&  v_4 
\end{pmatrix},
$$ where $v_i$ denotes vertex $i$. 
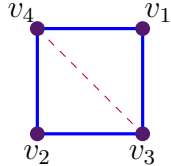
\begin{figure}[h!]
 \centering
   %  \includegraphics[width=0.3\textwidth]{Polytope1}
    % \vspace{-0.15in}
    \begin{tikzpicture}[scale=1.4]
\draw [purple,dashed](0,1) -- (1,0);
\draw [blue,very thick](1,1) -- (1,0);
\draw [blue,very thick](0,0) -- (0,1);
\draw [blue,very thick](0,0) -- (1,0);
\draw [blue,very thick](1,1) -- (0,1);

\node at (1,-.2) {$v_3$};
\node at (-0.15,1.15) {$v_4$};
\node at (1.15,1.15) {$v_1$};
\node at (0,-0.2) {$v_2$};
\fill[pur1] (0,1) circle[radius=2pt];
\fill[pur1] (1,0) circle[radius=2pt];
\fill[pur1] (1,1) circle[radius=2pt];
\fill[pur1] (0,0) circle[radius=2pt];
\end{tikzpicture}
     \caption{The polygon $P  = {\rm conv}(A)$. \label{fig:McKeithan}
     }
\end{figure}

This polytope $P$ (see Figure \ref{fig:McKeithan}) is smooth (i.e.~the associated toric variety is smooth). Hence $V_j$, the sum of all Chern-Mather volumes of all faces of dimension $j$, is equal to the sum of all normalized volumes of dimension $j$ faces. The normalized volume of $P$ is $V_2={\rm Vol}(P)=2$, there are four edges (dimension one faces) so $V_1=4$, and there are four vertices giving $V_0=4$. % Applying Theorem 1.2 of \cite{HS} gives polar degrees $\delta_i=\delta_i(X_{A})$, $$
%(\delta_0,\delta_1, \delta_2)=(2,2,2).
%$$
Applying \eqref{eq:computeEDToric} gives $$ {\rm EDdegree}(X_{A})=\sum_{j=0}^2(-1)^{2-j}(2^{j+1}-1)\cdot V_j=6.$$
\end{proof}

\begin{figure}
\centering
\resizebox{0.8\textwidth}{!}{%
\begin{minipage}{0.5\textwidth}
\includegraphics[width=\textwidth]{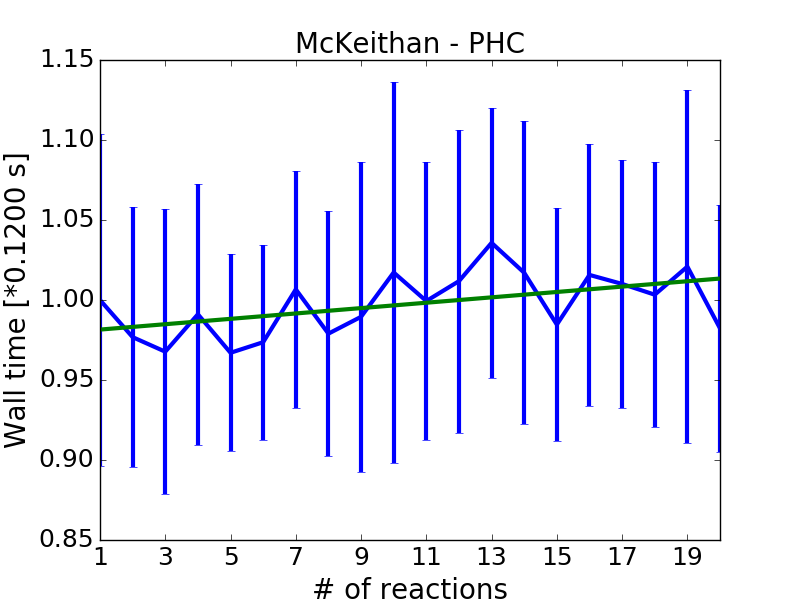}
\subcaption{Direct computation with PHCpack}
\end{minipage}%
\begin{minipage}{0.5\textwidth}
\includegraphics[width=\textwidth]{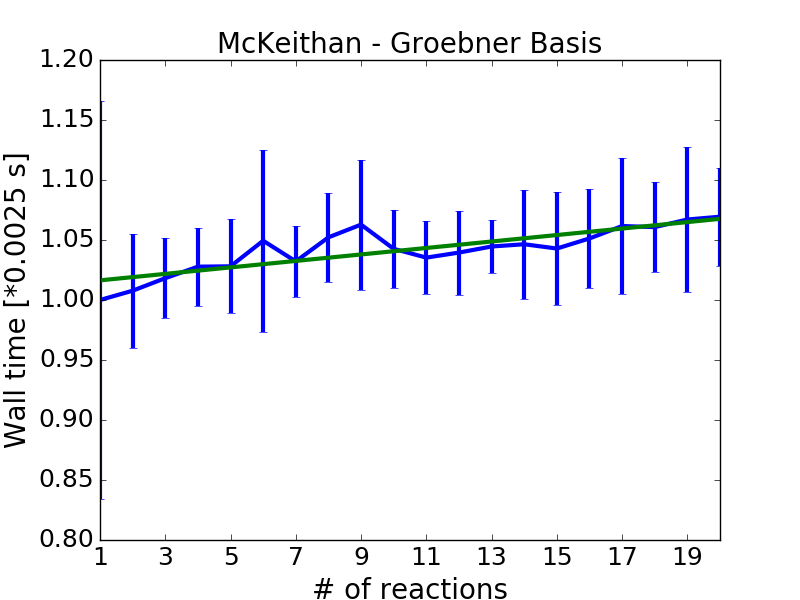}
\subcaption{Computitation of the Gr{\"o}bner basis}
\end{minipage}}
\caption{A comparison of the runtime for solving the ED problem for the McKeithan model. Both methods seem to have an approximately constant run time; however the Gr{\"o}bner basis algorithm outperforms the numerical algebraic geometry approach, even when we account for the fact that the triangular system still needs to be solved.}
\label{fig:McKTim}
\end{figure}

In a similar manner to the sequestration networks the ED problem can be solved numerically with PHCpack or by computing the Gr{\"o}bner basis of the model. Our findings, which are summarized in Figure \ref{fig:McKTim}, show a nearly constant computation time, the slight increase in time observed is likely due to extra computational overhead when working with polynomial systems in more variables.

\subsection{Multimeric Pore-Forming Toxins}\label{subsec:pore}

Pore forming toxins attack cells by assembling pores in the cell membrane from monomers. The pores cause the cells to leak, which eventually leads to cell death \cite{Voskobionik2006,Rosado2007,Iacovache2010}. A number of pore forming cytotoxins are employed by various bacteria such as Streptococcus pneumonia, Staphylococcus aureus, Escherichia coli and Myobacterium tuberculosis, but they are also used in the eurokaryotic immune system to kill pathogens and infected cells \cite{Los2013}.

For our calculation we adapt the $N$-monomer pore model from \cite{Lee2016} which follows the reaction scheme
\begin{align}
X_1 + X_i &\xrightleftharpoons[l_i]{k_i} X_{i+1},\;\;\; 1\geq i \geq N-2,\nonumber\\
X_1 + X_{N-1} &\xrightarrow{k_0} X_N.
\label{sch:poreOrig}
\end{align}
In the current form the non-zero closure of the steady state variety of \eqref{sch:poreOrig} is not a toric variety. To obtain a non-zero closure which is a toric variety we make a slight modification to the model. In our modified model we make the additional assumption that fully formed pores $X_N$ have a finite lifetime and can open up again. This change  results in the scheme given below
\begin{align}
X_1 + X_i &\xrightleftharpoons[l_i]{k_i} X_{i+1},\;\;\; 1\geq i \geq N-2,\nonumber\\
X_1 + X_{N-1} &\xrightleftharpoons[l_0]{k_0} X_N.
\label{sch:pore}
\end{align}
Since the ED degree is independent of the model parameters, \eqref{sch:poreOrig} can be thought of as \eqref{sch:pore} in the limit of $l_0 \rightarrow 0$ or the lifetime of the pore $l_0^{-1} \rightarrow\infty$, which gives us confidence that our modification still bears biological relevance.

From the reaction scheme we find the reaction matrix
\begin{equation}
\Gamma =
(
\underbrace{-e_1-e_1+e_2, -e_1-e_2+e_3, \cdots, -e_1 - e_{N-1} + e_N,}_{\Gamma_1} | -\Gamma_1 
).
\end{equation}
We also find the flux vector
\begin{equation}
R(x) =
\left(
k_1x_1x_1,
k_2x_1x_2,
\dots,
k_{N-2}x_1x_{N-2},
k_0x_1x_{N-1},
l_1x_2,
l_2x_3,
\dots,
l_{N-2}x_{N-1},
l_0x_N
\right)^T.
\end{equation}From the ODE system describing the dynamics of the model we can find the steady state variety. We proceed by finding the exponents of the parametrization of the steady state variables and computing the ED degree in Theorem \ref{theorem:pore}.

% \begin{lemma}
% Let $V_{\mathfrak{N}_{n}}$ be the steady state variety of the n-site pore network $\mathfrak{N}_{n}$. The non-zero closure of $V_{\mathfrak{N}_{n}}$ is the affine \textit{rational normal curve}; that is $V_{\mathfrak{N}_n}^{\neq 0}=X_{B_n}$ where $X_{A_n}$ is the affine toric variety defined by the $1 \times (n+1)$ integer matrix 
% $$
% B_n=\begin{pmatrix}
% n & n-1 & n-2 &n-3 &  \cdots & 2 &1 & 0 
% \end{pmatrix},
% $$ and, hence, $\dim(X_{B_n})=1$.
% \label{lem:pore}
% \end{lemma}
% \begin{proof}
% By induction and theorems from \cite{ProcPhos,CRSToric}.
% \end{proof}

\begin{theorem}
Let $V_{\mathfrak{N}}$ be the steady state variety of the $N$-site pore network $\mathfrak{N}$. Then we have that $${\rm EDdegree}\left( \overline{V_{\mathfrak{N}}^{\neq 0}}\right)=3N-2.$$
%Further $X_A^{\vee}$ is a hypersurface with degree $2N-2$.
\label{theorem:pore}
\end{theorem}\begin{proof}
Let $I_{\Gamma}$ be the ideal in $k[x_1,\dots,x_N]$ generated by $\Gamma R(x) $. Computing a Gr\"obner basis of $I_\Gamma$ (in the graded reverse lexographic order) we find that $I_{\Gamma}$ is a prime ideal generated by binomials. Finding a matrix with kernel given by the set of exponents of $I_\Gamma$ (interpreted as in \eqref{eq:BinomialIdeal}) we see that  $V_{\mathfrak{N}}$ is isomorphic to the affine \textit{rational normal curve} (which is an affine toric variety). That is, $V_{\mathfrak{N}}^{\neq 0}\cong X_{B}$ where $X_{B}$ is the affine toric variety defined by the $1 \times (N+1)$ integer matrix 
$$
B=\begin{pmatrix}
N & N-1 & N-2 &N-3 &  \cdots & 2 &1 & 0 
\end{pmatrix}.
$$ Taking the projective closure of $X_{B}$ reveals that $X_A$ is the projective toric variety defined by the $2 \times (N+1) $ integer matrix $$
A=\begin{pmatrix}
N & N-1 & N-2 &N-3 &  \cdots & 2 &1 & 0 \\
1 & 1 & 1 & 1 & \cdots & 1 & 1 & 1\\
\end{pmatrix},
$$ i.e.~$X_A$ is the projective rational normal curve. The polytope $P={\rm conv}(A)$ is one dimensional and consists of a line with vertices $0$ and $N$. This polytope is smooth, and has normalized volume $N$. Computing with this polytope we have that the Chern-Mather volumes of $P$ are $V_0=2$ (that is the number of vertices), and $V_1=N={\rm Vol}(P)$ (that is the normalized volume of $P$, i.e~the length of the line). Applying \eqref{eq:computeEDToric} gives $$ {\rm EDdegree}(X_{A})=3V_1-V_0=3N-2.$$
\end{proof}
%Applying Theorem 1.2 of \cite{HS} gives polar degrees $\delta_i=\delta_i(X_{A})$, $$
%(\delta_0,\delta_1)=(2V_1-V_0,V_1)=(2n-2,n).
%$$ 
\begin{remark}
It is shown in Example 1.3 of \cite{HS} that the usual Euclidean norm (i.e.~where we set $\lambda_i=1$ in \eqref{eq:opt3}) gives the generic value of $ {\rm EDdegree}(X_{A})=3N-2$ for the rational normal curve. 
\end{remark}

\begin{figure}
\centering
\includegraphics[width=.45\textwidth]{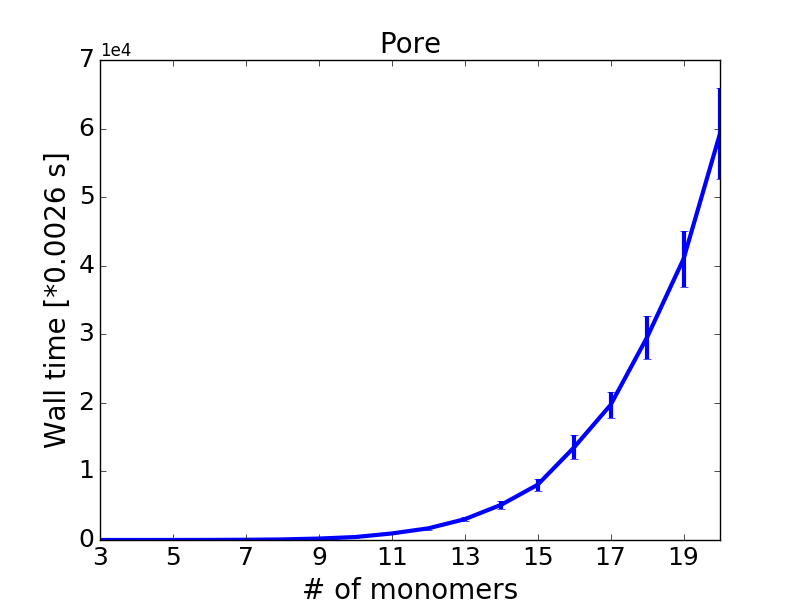}
\caption{The wall time of lex Gr{\"o}bner basis computation for the pore model. As expected from the ED degree the run time grows quickly, however, not linearly, with the number of monomers. %Therefore, the ED degree can predict a general tendency in the run time, but not its functional form.
}
\label{fig:PoreTim}
\end{figure}

In a similar manner to the distributive phosphorylation networks the computational time rises exponentially with the number of monomers in a formed pore, see Figure \ref{fig:PoreTim}.

\section{Discussion}\label{section:Biology}

Ultimately we would like to relate the algebraic complexity of our models back to biological features of the underlying reaction network, as there is no obvious connection between projective toric varieties and biology. Hence, in Table \ref{tab:Biology} we collected some common metrics used to classify chemical reaction networks such as their deficiency, their reversibility and their multistationarity. We do not have enough data to see an obvious pattern of how the functional form of the ED degree relates to these network features.

It seems that more insight can be gained when the combinatorics of the ``substrate" is considered. For the phosphorylation networks the substrates are proteins with $N$ sites which can be in a phosphorylated or unphosphorylated state. A priori, this gives $2^N$ possible states of the substrate. These are achieved in the distributive network due to the constant docking and undocking (1 docking event per phosphorylation) of the phosphatases and kinases. In the processive network, however, there only exist $N+1$ possible states, completely unphosphorylated, phosphorylated up to the $i^{th}$ site, and fully phosphorylated. It takes a minimum of one docking event to form the end product.
A similar situation to the distributive network is encountered in the pore forming network. From the ``seed monomer" there always exist $2$ possible docking sites until the $N-1^{st}$ docking event when the pore is formed, giving a total of $2^{N-2}$ possible states of the pore.
It is less clear what the ``substrate" for the sequestration network is. However, since the ultimate goal of such networks is to eliminate certain chemical species we takes the zero complex, $\emptyset$, as our product. It takes only one docking event to form the zero complex and there are $N-1$ reactions which can form it. The zero complex itself can only exist in one state, however, reminding ourselves that $\emptyset$ is just a modeling notation and that in the actual biological system the product of e.g. $X_1+X_2$ and $X_2+X_3$ may well be different we count the products of these reactions as different states. Hence, we define the combinatorial complexity of the sequestration network to be $N-1$.
The McKeithan model needs one docking event, $A+B\rightarrow X_1$, to form the initial product. The rest of the reaction network is comprised of dynamical steps only. However, in a similar fashion to the processive phosphorylation network, the ``substrate" $X_i$ can exist in $N$ different states and therefore the combinatorial complexity of the McKeithan model is $N$.

Despite the limited number of models under investigation we can conjecture from Table \ref{tab:Biology} that the functional form of the ED degree is intimately related to the combinatorial complexity of the ``substrates" and even scales directly with the minimum number of docking events it takes to form the final product.

\begin{table}[h!]
\centering
\resizebox{.95\linewidth}{!}{
\begin{tabular}{@{} *7c @{}}
\toprule 
 \multicolumn{1}{c}     {\color{Ftitle} Model}  &   {\color{Ftitle} Deficiency} &   {\color{Ftitle} Weakly reversible} &   {\color{Ftitle} Multistationary}  &   {\color{Ftitle} Combinatorial complexity}&   {\color{Ftitle} \# of docking events} &   {\color{Ftitle} EDdegree}\\ 
 \midrule 
  \color{line}Processive & \color{line}1 & \color{line}No& \color{line}No& \color{line}$N+1$& \color{line}1 & \color{line}$28$\\ 
 \color{line}Distributive & \color{line}variable & \color{line}No& \color{line}Yes& \color{line}$2^{N}$& \color{line}$N$ & \color{line}$23N+5$\\ 
  \color{line}Pore & \color{line}0 & \color{line}Yes& \color{line}No& \color{line}$2^{N-2}$& \color{line}$N-1$ & \color{line}$3N-2$\\ 
    \color{line}Sequestration & \color{line}0 & \color{line}No& \color{line}No& \color{line}$N-1$& \color{line}1 & \color{line}$1$\\ 
\color{line}McKeithan & \color{line}0 & \color{line}Yes& \color{line}No& \color{line}$N$& \color{line}1& \color{line}$6$\\
\bottomrule
 \end{tabular}}\vspace{1mm}
\caption{Biological and chemical reaction network theory metrics for our models.} \label{tab:Biology}
 \end{table}

\subsection{Restrictions and Extensions}

Throughout the previous sections we make assumptions which solely serve the purpose of illustrating the mathematical concepts behind our approach. However, in applications, especially when considering experimental data, these assumptions can be relaxed in a rigorous manner. In this section we briefly discuss how to extend our approach to multiple data points and how to address the common problem of unmeasurable concentrations in network models.

%\begin{remark}
In the previous sections we consider only one data point $u \in \RR^n$ for mathematical simplicity. However, our approach can easily be generalized to a discrete set of data points by solving the Euclidean distance problem separately for each data point. We then find a global, real, positive minimum point on the steady state variety for each data point $u_q$ and their Euclidean distance $\delta_q$. By applying an appropriate selection criterion such as $\text{max}_q\; \delta_q < \epsilon$ we can test the model for its validity. It is obvious that, unless the task is parallelized appropriately, the complexity of solving the ED problem is additive in the number of data points.
%\end{remark}

% \begin{remark}
In real biological systems it is often not possible to measure the concentrations of all the chemical species in the model. One way to adapt our calculations is to compute an appropriate elimination ideal of the steady state ideal, which corresponds to a projection of the steady state variety, and proceed using the variety generated by the elimination ideal.
% \end{remark}
% As noted in Remark \ref{remark:ED_Projection} above, it may often be the case that some of the model variables are not measurable in practice. 
More explicitly, suppose that $V_{\mathfrak{N}}$ in $\CC^n$ is the steady state variety of a chemical reaction network $\mathfrak{N}$ and further suppose that we may only measure $m$, $m<n$, of the $n$ coordinate values in the model $\mathfrak{N}$. Without loss of generality we may assume that the first $m$ coordinates of a data point $u\in \CC^n$ can be measured, while the remaining $n-m$ coordinates cannot. In this case we would instead study the Euclidean distance problem for (the Zariski closure of) the projection of the steady state variety onto the first $m$ coordinates, that is we would study the ED problem for $\overline{\pi(V_{\mathfrak{N}})}$ where $\pi:\CC^n\to \CC^m$ is given by $\pi:(x_1,\dots,x_n)\mapsto(x_1,\dots,x_m)$. In this case the following proposition tells us that the ED degree of the projected model will still be bounded by the ED degree of $V_{\mathfrak{N}}$.

\begin{proposition}
Let $V$ be a projective variety in $\PP^n$ and let $\pi$ denote the projection map $\pi:\PP^n \to \PP^m$ with $m<n$. Then we have that $$
{\rm EDdegree}\left(\overline{\pi(V)}\right)\leq {\rm EDdegree}(V).
$$\label{prop:EdOfProjectionBound}
\end{proposition}
\begin{proof}
Let $W$ be an arbitrary subvariety of $\PP^n$. By \cite[Theorem 5.4]{DHOST} we have that $$
{\rm EDdegree}(W)=\delta_0(W)+\cdots+\delta_n(W),
$$ where the $\delta_i$ are the polar degrees of $W$. Also note that the polar degrees are non-negative by defintion, i.e.~$\delta_i(W)\geq 0$ for all $i$. By \cite[Theorem 4.1]{piene1978polar} (and by the fact that the degree of the preimage of a projection map is equal to the degree of the map multiplied by the degree of the image, see, for example, \cite[Proposition 5.5]{Mumford}) we have that $$
\delta_i\left(\overline{\pi(V)} \right)\leq \delta_i(V), \;\;\; {\rm for \; all\; }i=0,\dots,n.
$$ Hence, putting this together, we have that \small$$
{\rm EDdegree}\left(\overline{\pi(V)} \right)=\delta_0\left(\overline{\pi(V)} \right)+\cdots+\delta_n\left(\overline{\pi(V)} \right) \leq \delta_0(V)+\cdots+\delta_n(V)={\rm EDdegree}(V).
$$\normalsize
\end{proof}

\section{Conclusion}\label{section:Conclusion}

In this paper we calculated the Euclidean Distance degree (ED degree) for a number of common chemical reaction networks with toric steady states.
We demonstrate that the ED degree can be used to quantify the overall algebraic complexity of finding points on model closest to an observed data point. In general such points will be complex, however, we prove that biologically viable (real and positive) points are always contained in the solution set of the ED problem for toric models.

We show that the ED degree is constant for many common biological models, namely, the processive multi-site phosphorylation network, the McKeithan kinetic proofreading network and sequestration networks. Thus, finding the global minimum to the optimization problem \eqref{eq:OpRe} should exhibit the same algebraic complexity, independent of the number of reactions. Similarly, we show that the ED degree increases linearly for distributive multi-site phosphorylation networks and pore forming networks.

To illustrate the meaning of algebraic complexity in practice we solved the ED problem computationally using Gr{\"o}bner bases and numerical algebraic geometry (NAG) techniques. We confirmed  that, indeed, the run time for computing the Gr{\"o}bner bases stays constant for models with constant ED degree.

Care needs to be taken, however, when relating the functional form of the ED degree of a model back to other measures of `model complexity' such as the capability of a model to have multiple positive steady states. Indeed, the distributive network has been shown to have the capability of multistationarity, whereas the processive model does not. However, the pore forming network does not have multiple positive steady states for any choice of parameters. A pattern can be seen when comparing the combinatorial complexity to the ED degree, however, no rigorous connection between the two is derived in this paper.

The combinatorial calculation of the ED degree for toric models can be many orders of magnitude faster than calculating its Gr{\"o}bner basis. Knowing this value can also be used to optimize specially designed NAG algorithms since we would know a priori the number of solutions paths which need to be tracked by these methods. Future work could include a software package to calculate the ED degree for a large number of biologically relevant networks. Furthermore, a more detailed study of non-toric models would be desirable.

\section*{Acknowledgements}
The authors would like to thank Elizabeth Gross for suggesting the topic of this paper. This work also greatly benefited from helpful discussions with Heather Harrington and Anne Shiu which began at the American Institute of Mathematics (AIM) in San Jose. MFA would like to thank the EPSRC for supporting this research through grant EP/G03706X/1. MH was supported by a NSERC postdoctoral fellowship during the preparation of this work. 
\bigskip

\begin{small}

\bibliographystyle{plain}

\begin{thebibliography}{10}
\setlength{\itemsep}{-0.8mm}

\bibitem{PerezMillan}
Millan, M. P., Dickenstein, A., Shiu, A., \& Conradi, C. (2012). Chemical reaction systems with toric steady states. Bulletin of mathematical biology, \textbf{74(5)}, 1027-1065.

\bibitem{Feinberg1987}
M.~Feinberg, {\em Chemical reaction network structure and the stability of complex isothermal reactors-I. The deficiency zero and deficiency one theorems}, Chem. Eng. Sci., \textbf{42(10)} (1987), 2229–2268. 

\bibitem{Feinberg1988}
M.~Feinberg, {\em Chemical reaction network structure and the stability of complex isothermal reactors-II. Multiple steady states for networks of deficiency one}, Chem. Eng. Sci., \textbf{43(1)} (1988), 1-25.

\bibitem{Joshi2014}
B.~Joshi, A.~Shiu, {\em A survey of methods for deciding whether a reaction network is multistationary},
Math. Model. Nat. Phenom., special issue on ``Chemical dynamics'', Vol.~10, pg.~47--67, 2015.

\bibitem{ConradiShiu2017} Conradi C, Shiu A. {\em Dynamics of post-translational modification systems: recent progress and future directions}. arXiv preprint arXiv:1705.10913. 2017.

\bibitem{Salazar2009}
C.~Salazar, T.~H{\"o}fer, {\em Multisite protein phosphorylation - from molecular mechanisms to kinetic models},
FEBS Jour., \textbf{276(12)} (2009), 3177–3198.

\bibitem{Kitano2002}
H.~Kitano,
{\em Systems Biology: A Brief Overview},
Science (2002), 1662-1664.

\bibitem{Gross2016b}
E.~Gross, H. A.~Harrington, Z.~Rosen, B.~Sturmfels, {\em Algebraic Systems Biology: A Case Study for the Wnt Pathway},Bull. Math. Biol., \textbf{78(1)} (2016), 21–51. 

\bibitem{Gross2016a}
E.~Gross, B.~Davis, K. L.~Ho, D. J.~Bates, H. A.~Harrington, {\em Numerical algebraic geometry for model selection and its application to the life sciences}, Jour. Roy. Soc. Int., \textbf{13(123)} (2016).

\bibitem{Burnham2002}
K. P.~Burnham, D. R.~Anderson,
{\em Model selection and multimodel inference: a practical information-theoretic approach},
New York, NY: Springer (2002).

\bibitem{Chamberlin1965}
T. C.~Chamberlin,
{\em The method of multiple working hypotheses: with this method the dangers of parental affection for a favorite theory can be circumvented},
Science \textbf{148} (1965), 754–759.

\bibitem{Kirk2013}
P. D. W.~Kirk, T.~Thorne, M. P. H.~Stumpf,
{\em Model selection in systems and synthetic biology},
Curr. Opin. Biotechnol. \textbf{24}(2013), 767–774.

\bibitem{DHOST}
J.~Draisma, E.~Horobe\c{t}, G.~Ottaviani, B.~Sturmfels and R.~Thomas:The Euclidean distance degree of an algebraic variety, {\em Foundations of~Computational~Mathematics}~{\bf 16} (2016) 99--149.

\bibitem{HS}
M.~Helmer and B.~Sturmfels:Nearest Points on Toric Varieties, {\em arXiv:1603.06544v2} 2016.

\bibitem{Craciun2009}
G.~Craciun, A.~Dickenstein, A.~Shiu, B.~Sturmfels,{\em Toric dynamical systems}Jour. Symb. Comp., \textbf{44(11)} (2009), 1551–1565.

\bibitem{FS1}
 S.~Friedland and M.~Stawiska: Best approximation on semi-algebraic sets and k-border rank approximation of symmetric tensors. arXiv:1311.1561.

\bibitem{FS2}
 S.~Friedland and M.~Stawiska: Some approximation problems in semi-algebraic geometry. Banach Center Publications. 2015;107(1):133-47.

\bibitem{Cohen2001}
P.~Cohen,{\em The role of protein phosphorylation in human health and disease},Eur. Jour. Biochem., \textbf{268(19)} (2001), 5001–10.

\bibitem{OSS} G.~Ottaviani, P-J.~Spaenlehauer, B.~Sturmfels. {\em Exact solutions in structured low-rank approximation}. SIAM Journal on Matrix Analysis and Applications. 2014 Dec 11;35(4):1521-42.

\bibitem{THP15} M.~Trager, M.~Hebert, J.~Ponce. The joint image handbook. In Proceedings of the IEEE international conference on computer vision 2015 (pp. 909-917).

\bibitem{SF17}A.~Stegeman, S.~Friedland. On best rank-2 and rank-(2, 2, 2) approximations of order-3 tensors. Linear and Multilinear Algebra. 2017 Jul 3;65(7):1289-310.
\bibitem{FKO17}G.~Fl{\o}ystad, J.~Kileel, G.~Ottaviani. {\em The Chow form of the essential variety in computer vision}. Journal of Symbolic Computation. 2017 Apr 5.
\bibitem{CNAS}M.~Compagnoni, R.~Notari, F.~Antonacci, A.~Sarti. A comprehensive analysis of the geometry of TDOA maps in localization problems. Inverse Problems. 2014 Feb 6;30(3):035004.
\bibitem{CCBAST}MA.~Compagnoni, A.~Canclini, PA.~Bestagini, FA.~Antonacci, A.~Sarti, S.~Tubaro. TDOA denoising for acoustic source localization. CoRR, vol. abs/1509.02380. 2015 Sep.

\bibitem{CH15} L.~Condat, A.~Hirabayashi. Cadzow denoising upgraded: A new projection method for the recovery of Dirac pulses from noisy linear measurements. Sampling Theory in Signal and Image Processing. 2015;14(1):p-17.

\bibitem{CFM17} M.~Casanellas, J.~Fern\'andez-S\'anchez J, M.~Michalek. Local equations for equivariant evolutionary models. Advances in Mathematics. 2017 Jul 31;315:285-323.

\bibitem{Roberts1997}
R.~Roberts, N. A.~Timchenko, J. W.~Miller, S.~Reddy, C. T.~Caskey, M. S.~Swanson, L. T.~Timchenko,{\em Altered phosphorylation and intracellular distribution of a (CUG)n triplet repeat RNA-binding protein in patients with myotonic dystrophy and in myotonin protein kinase knockout mice},PNAS, \textbf{94(24)} (1997), 13221–13226.

\bibitem{Gual2005}
P.~Gual, Y.~Le Marchand-Brustel, J.-F~Tanti,{\em Positive and negative regulation of insulin signaling through IRS-1 phosphorylation},Bioch., \textbf{87(1)} (2005), 99–109. 

\bibitem{Rubinstein}
B. Y.~Rubinstein, H. H.~Mattingly, A. M. Berezhkovskii, S. Y. Shvartsman,{Long-term dynamics of multisite phosphorylation},Mol Biol Cell. 27(14), 2331-40

\bibitem{FGLM}J. C.~Faugere, P.~Gianni, D.~Lazard, T.~Mora: Efficient Computation of Zero-dimensional Gr\"obner Bases by Change of Ordering, Journal of Symbolic Computation, Volume 16, Issue 4, 1993, Pages 329-344, ISSN 0747-7171.


\bibitem{Aubol2003}
B. E.~Aubol, S.~Chakrabarti, J.~Ngo, J.~Shaffer, B.~Nolen, X.-D.~Fu, G.~Gourisankar, J. A.~Adams,{\em Processive phosphorylation of alternative splicing factor/splicing factor 2},PNAS, \textbf{100(22)} (2003), 12601–6.

\bibitem{Ma2008}
C.-T.~Ma, A.~Velazquez-Dones, J. C.~Hagopian, G.~Ghosh, X.-D.~Fu, J. A.~Adams,{\em Ordered Multi-site Phosphorylation of the Splicing Factor ASF/SF2 By SRPK1},Jour. Mol. Biol., \textbf{376(1)} (2008), 55–68. 

\bibitem{Burack1997}
W. R.~Burack, T. W.~Sturgill,{\em The activating dual phosphorylation of MAPK by MEK is nonprocessive}, Biochem., \textbf{36(20)} (1997), 5929–5933.

\bibitem{ProcPhos}
C.~Conradi, and A.~Shiu: {\em A global convergence result for processive multisite phosphorylation systems}. Bulletin of mathematical biology, 77(1), pp.126-155, 2015.

\bibitem{Ferrell1997}
J. E.~Ferrell, R. R.~Bhatt,{\em Mechanistic studies of the dual phosphorylation of mitogen-activated protein kinase}, Jour. Biol. Chem., \textbf{272(30)} (1997), 19008–16. 

\bibitem{Waas2001}
W. F.~Waas, H. H.~Lo, K. N.~Dalby,{\em The kinetic mechanism of the dual phosphorylation of the ATF2 transcription factor by p38 mitogen-activated protein (MAP) kinase alpha. Implications for signal/response profiles of MAP kinase pathways},Jour. Biol. Chem., \textbf{276(8)} (2001), 5676–84.

\bibitem{Nash2001}
P.~Nash, X.~Tang, S.~Orlicky, Q.~Chen, F. B.~Gertler, M. D.~Mendenhall, F.~Sicheri, T.~Pawson, M.~Tyers,{\em Multisite phosphorylation of a CDK inhibitor sets a threshold for the onset of DNA replication}, Nature, \textbf{414(6863)} (2001), 514–521.

\bibitem{Deschaies2001}
R. J.~Deshaies, J. E.~Ferrell,{\em Multisite Phosphorylation and the Countdown to S Phase},Cell, \textbf{107(7)} (2001), 819–822.

\bibitem{CRSToric}
 M. P.~Mill\'an , A.~Dickenstein, A.~Shiu, C.~Conradi: Chemical reaction systems with toric steady states. Bulletin of mathematical biology. 2012 May 1;74(5):1027-65.

%\bibitem{Craciun2008}
%G.~Craciun, C.~Pantea,{\em Identifiability of chemical reaction networks},Jour. Math. Chem., \textbf{44(1)} (2008), 244–259. 

\bibitem{Bertini}
          D. J.~Bates, J.D.~Hauenstein, A. J.~Sommese,
          and Charles W. Wampler.
          \newblock Bertini: Software for Numerical Algebraic Geometry.
          \newblock Available at bertini.nd.edu with permanent doi: dx.doi.org/10.7274/R0H41PB5

\bibitem{Blanchard2004}
S. C.~Blanchard, R. L.~Gonzalez, H. D.~Kim, J. D.S.~Chu, Puglisi,{\em tRNA selection and kinetic proofreading in translation},Nat. Struct. Mol. Biol., \textbf{11(10)} (2004), 1008–1014. 

\bibitem{McKeithan1995}
T.W.~McKeithan,{\em Kinetic proofreading in T-cell receptor signal transduction},PNAS \textbf{92} (1995),5042-5046.

\bibitem{Hopfield1974}
J. J.~Hopfield,{\em Kinetic Proofreading: A New Mechanism for Reducing Errors in Biosynthetic Processes Requiring High Specificity},PNAS, \textbf{71(10)} (1974), 4135–4139.

\bibitem{Sontag2001}
E. D.~Sontag,{\em Structure and stability of certain chemical networks and applications to the kinetic proofreading model of T-cell receptor signal transduction},IEEE Trans. Aut. Cont., \textbf{46(7)} (2001), 1028–1047. 

\bibitem{Voskobionik2006}
I.~Voskoboinik, M. J.~Smyth, J. A.~Trapani,{\em Perforin-mediated target-cell death and immune homeostasis},Nat. Rev. Immun., \textbf{6(12)} (2006), 940–952. 

\bibitem{Rosado2007}
C. J.~Rosado, A. M.~Buckle, R. H. P.~Law, R. E.~Butcher, W.-T.~Kan, C. H.~Bird, … J. C.~Whisstock,{\em A Common Fold Mediates Vertebrate Defense and Bacterial Attack},Science, \textbf{317(5844)} (2007).

\bibitem{Iacovache2010}
I.~Iacovache, M.~Bischofberger, F. G.~van der Goot,{\em Structure and assembly of pore-forming proteins},Curr. Opin. Struct. Biol., \textbf{20(2)} (2010), 241–246.

\bibitem{Los2013}
F. C. O.~Los, T. M.~Randis, R. V.~Aroian, A. J.~Ratner,{\em Role of pore-forming toxins in bacterial infectious diseases},Microbiol. Mol. Biol. R., \textbf{77(2)} (2013), 173–207.

\bibitem{Lee2016}
A. A.~Lee, M. J.~Senior, M. I.~Wallace, T. E.~Woolley, I. M.~Griffiths,{\em Dissecting the self-assembly kinetics of multimeric pore-forming toxins}, Jour. Roy. Soc. Int., \textbf{13(114)} (2016).

\bibitem{GKZ}
I. M.~Gelfand, M.~Kapranov and A.~Zelevinsky. {\em Discriminants, Resultants, and Multidimensional Determinants}, Birkh\"auser, Boston, 1994.

% \bibitem{piene1988cycles}
% R.~Piene: Cycles polaires et classes de Chern pour les vari{\'e}t{\'e}s projectives singuli{\`e}res,{\em Introduction \`a la th\'eorie des singularit\'es},  Travaux en Cours, {\bf 37}, 7--34, Hermann, Paris, 1988.  

\bibitem{RealSols} F.~Sottile: {\em Real solutions to equations from geometry}. Vol. 57, 2011. American Mathematical Society.

\bibitem{GBCP}
B.~Sturmfels:{\em Gr\"obner Bases and Convex Polytopes}, University Lecture Series, Vol 8, American Mathematical Society, Providence, RI, 1996.

\bibitem{Mumford}
D.~Mumford: {\em Algebraic Geometry: Complex projective varieties}. vol. 1. Springer Science \& Business Media; 1995 Feb 15.

\bibitem{piene1978polar}
R.~Piene:  Polar classes of singular varieties,{\em Annales Scientifiques de l'\'Ecole Normale Sup\'erieure} {\bf 11} (1978) 247--276.  

% \bibitem{AluffiCM}
%  P.~Aluffi: Projective duality and a Chern-Mather involution,{\tt arXiv:1601.05427v2}.

%\bibitem{Sung01051997}
%P.~Sung: {\em Yeast Rad55 and Rad57 proteins form a heterodimer that functions with replication protein A to promote DNA strand exchange by Rad51 recombinase},Genes \& Development Vol. \textbf{11} (1997) 1111-1121.

% \bibitem{BHSW}
% D.~Bates,  J.~Hauenstein, A.~Sommese and C.~Wampler:{\em  Numerically Solving Polynomial Systems with Bertini}, Software, Environments and Tools, Vol 25, SIAM, Philadelphia, PA, 2013

\bibitem{Velazquez2005}
A.~Velazquez-Dones, J. C.~Hagopian, C.-T.~Ma, X.-Y.~Zhong, H.~Zhou, G.~Ghosh, X.-D.~Fu, J. A.~Adams,{\em Mass spectrometric and kinetic analysis of ASF/SF2 phosphorylation by SRPK1 and Clk/Sty},Jour. Biol. Chem., \textbf{280(50)} (2005), 41761–8. 

\bibitem{Yu2018}
Yu, P. Y., \& Craciun, G. (2018). {\em Mathematical Analysis of Chemical Reaction Systems. Israel Journal of Chemistry}.

% \bibitem{Lorenz}
% B.~Lorenz: Classification of smooth lattice polytopes with few lattice points. arXiv preprint arXiv:1001.0514. 2010 (Also a thesis).  

% \bibitem{Zie}
% G.~Ziegler:{\em Lectures on Polytopes},Graduate Texts in Math.~{\bf 152}, Springer, New York,1995.

\bibitem{Aoki2011}
K.~Aoki, M.~Yamada, K.~Kunida, S.~Yasuda,M.~Matsuda,{\em Processive phosphorylation of ERK MAP kinase in mammalian cells},PNAS, \textbf{108(31)} (2011), 12675–80.

% \bibitem{fulton}
% W.~Fulton: {\em Intersection Theory}, Springer, Berlin, 2nd edition, 1998. 

\bibitem{PSAlgStat}L.~Pachter and B.~Sturmfels. Algebraic statistics for computational biology. Cambridge University Press, 2005.

% \bibitem{M2}
%  D.~Grayson and M.~Stillman:{\em Macaulay2, a software system for research in algebraic geometry}, available at {\tt www.math.uiuc.edu/Macaulay2/}.

\bibitem{Holstein}
K.~Holstein, D.~Flockerzi, C.~Conradi,{Multistationarity in sequential distributed multisite phosphorylation networks},Bull. Math. Biol. 75(11), 2028-2058.

\bibitem{FultonToric}
W.~Fulton: {\em Introduction to Toric Varieties}. Princeton University Press; 1993.

% \bibitem{ottaviani2014exact}
% G.~Ottaviani, P-J.~Spaenlehauer and B.~Sturmfels: Exact solutions in structured low-rank approximation, {\em SIAM Journal on Matrix Analysis and Applications}  {\bf 35} (2014) 1521--1542.  

\bibitem{CLS}
D.~Cox, J.~Little and H.~Schenck: {\em Toric Varieties},
 Graduate Studies in Mathematics, 
 Volume 124, American Mathematical Society, Providence, RI, 2011.

\bibitem{SageMath}
The Sage Developers,
SageMath, the Sage Mathematics Software System (Version 7.2), 2016, http://www.sagemath.org.

\bibitem{PHC}
J.~Verschelde,
Algorithm 795, PHCpack: A general-purpose solver for polynomial systems by homotopy continuation,
ACM Transactions on Mathematical Software, 25(2):251--276, 1999.

\bibitem{Mahadevan2002}
R.~Mahadevan, J.S.~Edwards, F.J.~Doyle,
Dynamic flux balance analysis of diauxic growth in escherichia coli, Biophys. Jour., {\em 83(3)} (2002), 1331–1340.

\bibitem{Gunawardena2009}
M.~Thomson, J.~Gunawardena,
The rational parameterisation theorem for multisite post-translational modification systems,
Jour. Theor. Biol., \textbf{261(4)} (2002), 626–636.

\bibitem{Chellaboina2009}
V.~Chellaboina, S.~Bhat, W.~Haddad, D.~Bernstein,
Modeling and analysis of mass-action kinetics,
IEEE Control Systems, \textbf{29(4)} (2009), 60–78. 

\bibitem{Gunawardena2003}
J.~Gunawardena,
Chemical reaction network theory for in-silico biologists,
Preprint, (2003).

\bibitem{Michaelis1913}
L.~Michaelis,  M.L.~Menten,
Die Kinetik der Invertinwirkung,
Biochem Z, \textbf{49} (1913), 333–369.
\end{thebibliography}

\end{small}

\bigskip

\noindent
\footnotesize {\bf Authors' addresses:}

\smallskip

\noindent
Michael Adamer:
Mathematical Institute, University of Oxford,
Woodstock Rd, Oxford, OX2 6GG, UK\\
{\tt adamer@maths.ox.ac.uk}

\noindent 
Martin Helmer:
Department of Mathematical Sciences, University of Copenhagen, 
Universitetsparken 5, DK-2100 Copenhagen, Denmark \\
{\tt m.helmer@math.ku.dk}

\end{document}